%% file: main.tex
\documentclass[journal]{IEEEtran}

\usepackage{amsmath}    % need for subequations
\usepackage{graphicx}   % need for figures24
\usepackage{epstopdf}
\usepackage{verbatim}   % useful for program listings
\usepackage{color}      % use if color is used in text8
\usepackage{subfigure}  % use for side-by-side figures
\usepackage{algorithm}
\usepackage{array}
\usepackage{algorithmic}
\usepackage{nomencl}
\usepackage{graphicx}
\usepackage{amssymb}
\usepackage{enumerate}
\usepackage{amsthm}
\usepackage{multirow}
\usepackage{listings}
\usepackage{url}
\usepackage[hidelinks]{hyperref}
\usepackage{rotating}
\usepackage{wasysym}
\usepackage{footmisc}

\newtheorem{example}{Example}

\newtheorem{lemma}{Lemma}

\newtheorem{theorem}{Theorem}

\newtheorem{program}{Program}
\newtheorem{procedure}{Procedure}

\newcommand{\CASE}[1]{\STATE \textbf{case} #1\textbf{:} \begin{ALC@g}}			
\newcommand{\ENDCASE}{\end{ALC@g}\STATE \textbf{end case}}			
			
\newcommand{\DEFAULT}{\STATE \textbf{default:} \begin{ALC@g}}			
\newcommand{\ENDDEFAULT}{\end{ALC@g}}			
\newcommand{\DEFAULTLINE}[1]{\STATE \textbf{default:} }

\definecolor{mygreen}{rgb}{0,0.6,0}

\newif\ifmodify
\modifyfalse

\ifmodify
\newcommand{\modify}[1]{{\color{blue}{#1}}}
\else
\newcommand{\modify}[1]{#1}
\fi

\def\etal{\textit{et al.}}

\def\toolname{Y2U}
\def\yakindu{Yakindu}
\def\uppaal{UPPAAL}

\input{"macro/defs.tex"}

\begin{document}
	
\title{Design Verifiably Correct Model Patterns to Facilitate Modeling Medical Best Practice Guidelines with Statecharts\\ (Technical Report)}

\author{
	\IEEEauthorblockN{Chunhui Guo\textsuperscript{1}, Zhicheng Fu\textsuperscript{1}, Zhenyu Zhang\textsuperscript{2}, Shangping Ren\textsuperscript{2,1}, Lui Sha\textsuperscript{3}}
	
	\IEEEauthorblockA{		
		\textsuperscript{1}Department of Computer Science, Illinois Institute of Technology, Chicago, IL 60616, USA \\
		\textsuperscript{2}Department of Computer Science, San Diego State University, San Diego, CA 92182, USA \\
		\textsuperscript{3}Department of Computer Science, University of Illinois at Urbana-Champaign, Urbana, IL 61801, USA \\
		\{cguo13, zfu11\}@hawk.iit.edu, \{zzhang4430, sren\}@sdsu.edu, lrs@illinois.edu
	}
}

\maketitle

\begin{abstract}
	\input{abs}
\end{abstract}


\section{Introduction and Related Work}
\label{sec:intro}
\input{intro}

\section{Preliminary Work}
\label{sec:prelim}
\input{prelim}

\section{Model Pattern Design}
\label{sec:pattern}
\input{pattern}

\section{Model Pattern Correctness Proof}
\label{sec:proof}
\input{proof}

\section{Case Study}
\label{sec:exp}
\input{exp}

\section{Discussion}
\label{sec:discussion}
\input{discussion}

\section{Conclusion}
\label{sec:conclusion}
\input{conclusion}

\appendices
\input{appendix}

\section*{Acknowledgment}
\label{sec:acknowledgement}
\input{acknowledgement}

\bibliographystyle{ieeetr}
\bibliography{ref}

\end{document}

%% file: macro/defs.tex
\def\mycolor{black}

\def\kwif{{\color{\mycolor}\texttt{if}\ }}
\def\kwskip{{\color{\mycolor}\texttt{skip}\ }}
\def\kwthen{{\color{\mycolor}\texttt{then}\ }}
\def\kwelse{{\color{\mycolor}\texttt{else}\ }}
\def\kwfi{{\color{\mycolor}\texttt{fi}\ }}
\def\kwwhile{{\color{\mycolor}\texttt{while}\ }}
\def\kwdo{{\color{\mycolor}\texttt{do}\ }}
\def\kwod{{\color{\mycolor}\texttt{od}\ }}

\def\kw?{{\color{\mycolor}\texttt{?}\ }}

\def\triple#1#2#3{\{#1\}#2\{#3\}}

%% file: abs.tex
Improving patient care safety is an ultimate objective for medical cyber-physical systems. A recent study shows that the patients' death rate can be significantly reduced by computerizing medical best practice guidelines. To facilitate the development of computerized medical best practice guidelines, statecharts are often used  as a modeling tool because of  their high resemblances to disease and treatment models and their capabilities to provide rapid prototyping and simulation for clinical validations.
However, some implementations of statecharts, such as \yakindu\ statecharts,
are priority-based and have synchronous execution semantics which makes it difficult to model certain functionalities that are essential in modeling medical guidelines, such as two-way communications and configurable execution orders. Rather than introducing new statechart elements or changing the statechart implementation's underline semantics, we use existing basic statechart elements to design model patterns for the commonly occurring issues.
In particular, we show the design of model patterns for two-way communications and configurable execution orders and formally prove the correctness of these model patterns.
We further use a simplified airway laser surgery scenario as a case study to demonstrate how the developed model patterns address the two-way communication and configurable execution order issues and their impact on validation and verification of medical safety properties.

%% file: intro.tex
\IEEEPARstart{A}{} study shows that the patients' death rate can be significantly reduced by computerizing medical best practice guidelines~\cite{Mckinley2011computer}.
\modify{
Developing computerized disease and treatment models from medical best practice handbooks needs close interactions with medical professionals.
In addition, to satisfy the safety and correctness requirements, the derived models also need to be clinically validated and formally verified.
Over past two decades, many computer executable medical best practice
guideline models are developed, such as Asbru~\cite{Balser2002Asbru},
GLIF~\cite{patel1998representing}, GLARE~\cite{Terenziani2004GLARE},
EON~\cite{Tu2001EON}, and PROforma~\cite{fox1998disseminating}.
Rahmaniheris \etal~\cite{Rahmaniheris2016CBMS} have developed an
organ-centric approach to model medical best practice guidelines
with statecharts, which enable rapid prototyping and allow quick
clinical validations by medical staff through simulations.
Tan \etal~\cite{Tan2013DSN} proposed a design pattern for wireless
medical Cyber-Physical Systems (CPS). Jiang \etal~\cite{Jiang2015TIE}
presented a timed automata and synchronous dataflow based framework
for system design. Guo \etal~\cite{Guo2017ICCAD,Guo2018ICCPS} proposed
approaches to model and integrate medical resource availability and
relationships in existing medical guideline models.
To help improve clinical validation, Wu \etal\ have developed a workflow adaptation protocol~\cite{WuWorkflow2015} to help physicians safely adapt workflows to react to patient adverse events and a treatment validation protocol~\cite{WuTreatment2014} to enforce the correct execution sequence of performing a treatment based on preconditions validation, side effects monitoring, and expected responses checking. Both the workflow adaptation and treatment validation protocol are based on  pathophysiological models. In addition, based on organ-specific physiology, a system that integrates medical devices into semi-autonomous clusters in a network-fail-safe manner  has also been developed by Kang \etal~\cite{Kang2013CBMS}.
Christov \etal~\cite{Christov2013SEHC} proposed an approach to detect whether the performed
medical procedures have deviated from the recommended ways to perform the medical
procedures, i.e., medical best practices.
To formally verify safety properties of medical guideline models,
Guo \etal~\cite{Guo2016ICCPS} presented an approach to transform
statecharts to timed automata. Runtime verification techniques~\cite{Jiang2016TII,Guo2017COMPSAC}
were also applied to improve safety of medical guideline
systems at code level.
}

Most existing medical best practice guidelines in
hospital handbooks are represented by flowcharts~\cite{Guideline2015}
which are very similar to statecharts~\cite{harel1987statecharts},
so are many medical disease models and treatment models.
In addition to the high similarities between medical models and statecharts,
statecharts are executable and have become a widely used model in designing complex systems, 
such as avionics~\cite{Romdhani1995Avionics}, air traffic control systems~\cite{Whittle2005}, and medical systems~\cite{Rahmaniheris2016CBMS,Rahmaniheris2017ICCPS,Guo2016ICCPS}. These distinguishing features of statecharts have inspired us to use it as a computerized representation for medical best practice guidelines. However, there are functionalities that are essential in medical operations, such as two-way communications and configurable execution orders, which are not directly supported by some open source statechart modeling tools, such as \yakindu. We use the following two examples to illustrate the need of two-way communications and configurable execution orders in medical domain. By two-way communications we mean two statecharts can communicate with each other, and by configurable execution orders we mean the execution orders can be configured by users without change the model itself.

\begin{example}
	\label{ex:laser}
	Laser surgery~\cite{LaserSurgery} is a surgical procedure that uses a laser to remove problematic tissues and is widely used in airway surgery, thoracic surgery, eye surgery, etc.	
	For airway laser surgery, there are two potential dangers:
	(1) an accidental burn if both laser and ventilator are activated;
	and (2) a low-oxygen shock if the Saturation of Peripheral Oxygen ($\mathtt{SpO}$) level of the patient decreases below a given threshold (assume $95\%$)~\cite{Kim2010ICCPS}.
	To prevent the potential dangers, the airway laser and the ventilator must be able to communicate with each other bidirectionally, i.e., the communication needs to be two-way. In particular, when the surgery starts, the airway laser turns on and notifies the ventilator to turn off; and when the patient's $\mathtt{SpO}$ level becomes below $95\%$, the ventilator turns on and notifies the airway laser to turn off.
\end{example}

\begin{example}
	\label{ex:cough}
	In the chronic cough treatment guideline~\cite{Benich2011CoughGuideline}, chest radiography	has to be performed before sputum test and bronchoscopy, but the
	guideline does not specify the order between sputum test and bronchoscopy.
	The physicians can choose the medical procedure order
	$\mathtt{chest \ radiography}$ $\prec$ $\mathtt{sputum \ test}$ $\prec$ $\mathtt{bronchoscopy}$	
	or $\mathtt{chest \ radiography}$ $\prec$ $\mathtt{bronchoscopy}$ $\prec$ $\mathtt{sputum \ test}$	
	based on their experiences/preferences and medical resource availability.	
\end{example}

\modify{
Without the two-way communication support, the laser and the ventilator
in the airway laser surgery can be both activated and cause surgery fire.
We present more details about the scenario in the case study, i.e., Section~\ref{sec:exp}.
Without the configurable execution order support, if a physician deviate from
the medical procedure execution order specified in a medical
guideline system, the system can not identify whether or not the deviation is safe.
Hence, supporting two-way communications and configurable execution orders
are very important for modeling medical best practice guidelines.
}

However, most existing medical guideline modeling languages,
such as Asbru~\cite{Balser2002Asbru}, GLIF~\cite{patel1998representing},
and PROforma~\cite{Fox1998PROforma},
do not support these essential functionalities.
Althouh Simulink Stateflow~\cite{stateflow}, a statechart variant from
Matlab~\cite{matlab}, supports two-way communications through
\textit{condition actions} and preemptive execution semantics, it does not support execution order change without modifying existing models.

Open source \yakindu\ statecharts are priority-based and have synchronous execution semantics.
With such execution semantics, only higher priority statecharts can
send \textit{events} to lower priority statecharts, but not the other way around.
In addition, each statechart is pre-assigned an unique priority
to determine its execution order when a model is established.

A naive approach to implement two-way communication functionality is 
to use global variables shared among multiple statecharts. However, using global variables has known disadvantages of increased difficulty and complexity to maintain data consistence. An alternative is to take a similar approach as in Stateflow~\cite{stateflow} by introducing new elements into \yakindu\ to interrupt the execution
of the \textit{event} sender statechart and resume after handling the \textit{event}. 
However, such  approach changes \yakindu\ statecharts' underline
execution semantics and  violates \yakindu\ statecharts' original design goal.

To address the configurable execution order issue, a straight forward
approach is to customize medical guideline statechart models based
on execution orders provided by physicians.
However, the approach faces the following challenges:
(1) a medical guideline model requires a variant for every
physician to facilitate different execution orders, which is
cumbersome for practical use;
and (2) engineers need to be involved in clinical care to manually modify
medical guideline models when physicians change execution orders.

The major challenge of implementing the essential functionalities
in modeling medical guidelines, such as two-way communications
and configurable execution orders,
is that it has to be effortless for medical professionals
to validate its correctness and formally verifiable at reasonable cost.
The paper presents an approach to apply model patterns to support these
essential functionalities. The model patterns do not introduce new
statechart elements nor change statecharts' underline execution semantics.
Hence, the model patterns do not require additional effort for medical
professionals to validate their correctness.
In addition, our previous work~\cite{Guo2016ICCPS}
can also be applied to formally verify medical guideline models
that apply the model patterns.

The rest of the paper is organized as follows.
Section~\ref{sec:prelim} briefly introduce the preliminary
work on developing verifiably safe medical best practice
guidelines with statecharts.
We design the two-way communication model pattern and
configurable execution order model pattern in Section~\ref{sec:pattern}
and formally prove the model patterns' correctness in Section~\ref{sec:proof}.
A case study of a simplified airway laser surgery scenario
is performed in Section~\ref{sec:exp}. We present some discussions
in Section~\ref{sec:discussion} and conclude in
Section~\ref{sec:conclusion}.

%% file: prelim.tex
Our previous work~\cite{Guo2016ICCPS} presents an approach to
build verifiably safe executable medical guideline models in
two steps:
(1) use statecharts~\cite{harel1987statecharts} to model medical guidelines and
interact with medical professionals to validate the correctness of the medical
guideline models;
and (2) automatically transform medical guideline statecharts to
timed automata~\cite{behrmann2004tutorial} by the developed \toolname\
tool~\cite{Guo2016ICCPS} to formally verify safety properties.
In this section, we use the simplified airway laser surgery scenario in Example~\ref{ex:laser}
as an example to briefly summarize the process of building verifiably
safe executable medical guideline models.

\subsection{Model Medical Guidelines with Statecharts}
\label{subsec:model}

We use \yakindu\ statecharts to model the simplified airway laser surgery,
as shown in Fig.~\ref{fig:laser}.
Both laser and ventilator are modeled as a statechart with three states
($\mathtt{On}$, $\mathtt{Off}$, and $\mathtt{Syn}$) to represent the
devices' operation status.
To prevent the accidental burn danger caused by laser and ventilator are both
activated, we add the $\mathtt{Syn}$ state to ensure that
the laser/ventilator's activation procedure is delayed one step after
the ventilator/laser's deactivation procedure.
The initial state of the airway laser surgery system is: the laser
is off and the ventilator is on to supply oxygen to the patient.
When the ventilator is turned off, we assume that the patient's
$\mathtt{SpO}$ level decreases by 1 every second.

\setlength{\intextsep}{5pt}
\setlength{\columnsep}{0pt}
\begin{figure}[ht]
	\centering
	\includegraphics[width = 0.45\textwidth]{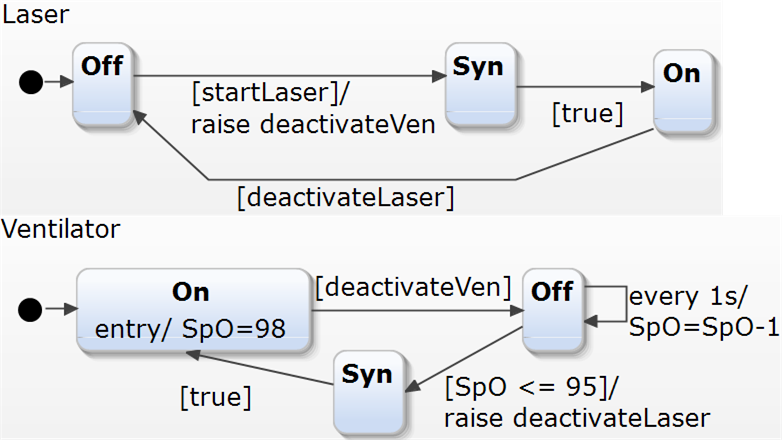}
%	\includegraphics[width = 0.8\textwidth]{fig/laser}
	%	\vspace{-1em}
	\caption{Simplified Airway Laser Surgery Statechart Model}
	\label{fig:laser}
\end{figure}

In the statechart shown in Fig.~\ref{fig:laser}, when an operating
surgeon sends the $\mathtt{startLaser}$ \textit{event}
to the $\mathtt{laser}$ to operate the surgery, the $\mathtt{laser}$
sends the $\mathtt{deactivateVen}$ \textit{event} to stop oxygen supply.
If the patient's $\mathtt{SpO}$ level reduces below $95\%$,
the $\mathtt{ventilator}$ turns on to supply oxygen and
sends the $\mathtt{deactivateLaser}$ \textit{event} to
stop the $\mathtt{laser}$.

\subsection{Verifiably Safe Statecharts}
\label{subsec:Y2U}

Statecharts contain basic elements, such as
\textit{states} and \textit{transitions}, and advanced elements,
such as \textit{composite states}. Although these advanced elements
provide modeling convenience, they increase the difficulty of both
clinical validation and formal verification.
As stated in the literature, one of the keys to achieving
system safety at reasonable cost is a serious and sustained
commitment to simplicity~\cite{Jackson2007Book,Carayon2011Handbook}.
To reduce the difficulty in both clinical validation and formal
verification,
our previous work~\cite{Guo2017CBMS} proposes a pattern-based statechart modeling approach
to model medical guidelines with basic statechart
elements and model patterns which are built upon these basic
elements to implement advanced statechart elements.

To formally verify medical guideline statechart models,
our previous work~\cite{Guo2016ICCPS} presents an approach to
transform statechart models to timed automata.
There are three key differences between \yakindu\ statecharts
and \uppaal\ timed automata:
(1) syntactic difference: \yakindu\ statecharts have some elements that are not
directly supported by \uppaal\ timed automata, such as \textit{event} and
\textit{timing trigger};
(2) execution semantics difference: \yakindu\ model is deterministic
and has synchronous execution semantics while the execution of \uppaal\ model is non-deterministic
and asynchronous;
and (3) simultaneous \textit{events} difference: \yakindu\ supports simultaneous events while \uppaal\ does not.

Our transformation handles above three key
differences by following approaches:
(1) syntactic difference: define 5 transformation rules for basic \yakindu\ statecharts
elements, i.e., \textit{state}, \textit{transition}, \textit{state action}, \textit{event},
and \textit{timing trigger};
(2) execution semantics difference: design 2 transformation rules to implement
transition trigger determinism and statechart execution determinism, 
and use the lockstep method~\cite{lockstep} to force synchronous execution;
and (3) simultaneous \textit{events} difference: design an event stack to
simulate simultaneous \textit{event} mechanism in \uppaal\ timed automata.

To ensure that the formal verification results in timed automata
is consistent with the statecharts, we prove the transformation
correctness. The approach also provides the capability to trace back
paths that fail safety properties from timed automata to statecharts.
Although our approach only transforms basic \yakindu\ statechart
elements, it is sufficient to provide formal verification functionality
for medical guideline statecharts
because the advanced statechart elements can be represented by
basic elements~\cite{Guo2017CBMS}.

%% file: pattern.tex
Some statecharts have priority-based, deterministic,
and synchronous execution semantics, such as \yakindu\ statecharts.
The execution semantics make it difficult to model certain features that are essential
in modeling medical guidelines, such as two-way communications and
configurable execution orders. In this section, we design model patterns
to support two-way communications and configurable execution orders
without changing statecharts' underline execution semantics
nor introducing new statechart elements.

\subsection{Design Model Pattern for Two-Way Communication}
\label{subsec:TWC}

In statechart models with multiple statecharts, two-way communications
are essential. For instance, the $\mathtt{laser}$ statechart and
$\mathtt{ventilator}$ statechart in Fig.~\ref{fig:laser} require
two-way communications.
We represent the two-way communication requirement by the statechart model shown in
Fig.~\ref{fig:TWCex1}. The two-way communication feature means that
statechart $\mathtt{S1}$ can receive \textit{event} $\mathtt{EB}$ raised
by $\mathtt{S2}$ and 
statechart $\mathtt{S2}$ can receive \textit{event} $\mathtt{EA}$ raised
by $\mathtt{S1}$.
However, if statechart $\mathtt{S1}$ has higher priority than statechart
$\mathtt{S2}$, the \textit{event} $\mathtt{EB}$ can not be passed from
$\mathtt{S2}$ to $\mathtt{S1}$ due to priority-based and synchronous
execution semantics.

\setlength{\intextsep}{5pt}
\setlength{\columnsep}{0pt}
\begin{figure}[ht]
	\centering
	%	\vspace{-2.5em}
	\includegraphics[width = 0.25\textwidth]{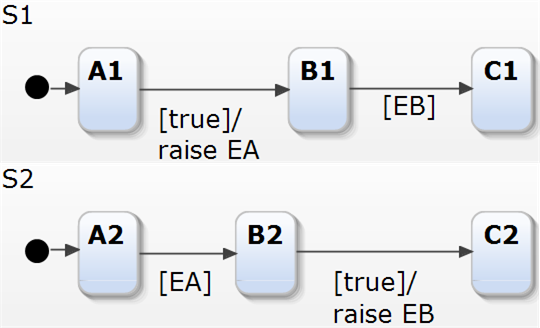}
%	\includegraphics[width = 0.6\textwidth]{fig/TWCex1}
	%	\vspace{-1em}
	\caption{Two-Way Communication Feature}
	\label{fig:TWCex1}
\end{figure}

In this section, we design a model pattern to support two-way communication in statecharts.
The model pattern design needs to satisfy the following two criteria:
(1) it does not changing statecharts' underline execution semantics;
and
(2) it does not introduce new statechart elements.
The reason of these two criteria is that new statechart
elements and new execution semantic rules increase the difficulty
for medical professionals to validate medical guideline models.
Based on the model pattern design criteria, our strategy to support
two-way communication is to queue all raised \textit{events}, add
\textit{logic execution cycles}, and re-raise the queued \textit{events}
in the added \textit{logic cycles}.
Note that each added \textit{logic execution cycle} takes one CPU time unit which is negligible
compared to clock time.
Under the two-way communication model pattern,
the statechart $\mathtt{S2}$ in Fig.~\ref{fig:TWCex1} can send
\textit{event} $\mathtt{EB}$ to statechart $\mathtt{S1}$.

We take \yakindu\ statecharts as an example to 
implement the two-way communication model pattern.
In particular,
for each execution cycle (we call it \textit{normal cycle}), we queue all
raised events.
After each normal execution cycle, we add $n-1$
\textit{logic cycles} to re-raise queued events, where
$n$ is the number of statecharts.
In each \textit{logic cycle}, each queued event is only visible to statecharts
whose priorities are higher than the event raiser, as lower
priority statecharts have already received the queued
event in the \textit{normal cycle}. Furthermore, at each \textit{logic cycle},
only those transitions that are triggered by queued events are
executed. This constraint is to ensure that the \textit{logic cycle}
does not change the model behavior other than facilitate higher
priority statecharts to receive events from lower priority statecharts.

The two-way communication model pattern contains
a $\mathtt{Manager}$ statechart and
a interface $\mathtt{TWC}$, as shown in Fig.~\ref{fig:TWCpattern}.
In particular, the $\mathtt{Manager}$ statechart has the
highest priority in the model and initializes the event queues.
The interface $\mathtt{TWC}$ declares four functions:
$\mathtt{initEventQueue()}$,
$\mathtt{push()}$,
$\mathtt{pop()}$,
and $\mathtt{isNormalExe()}$.

\setlength{\intextsep}{5pt}
\setlength{\columnsep}{0pt}
\begin{figure}[ht]
	\centering
	%	\vspace{-2.5em}
	\includegraphics[width = 0.49\textwidth]{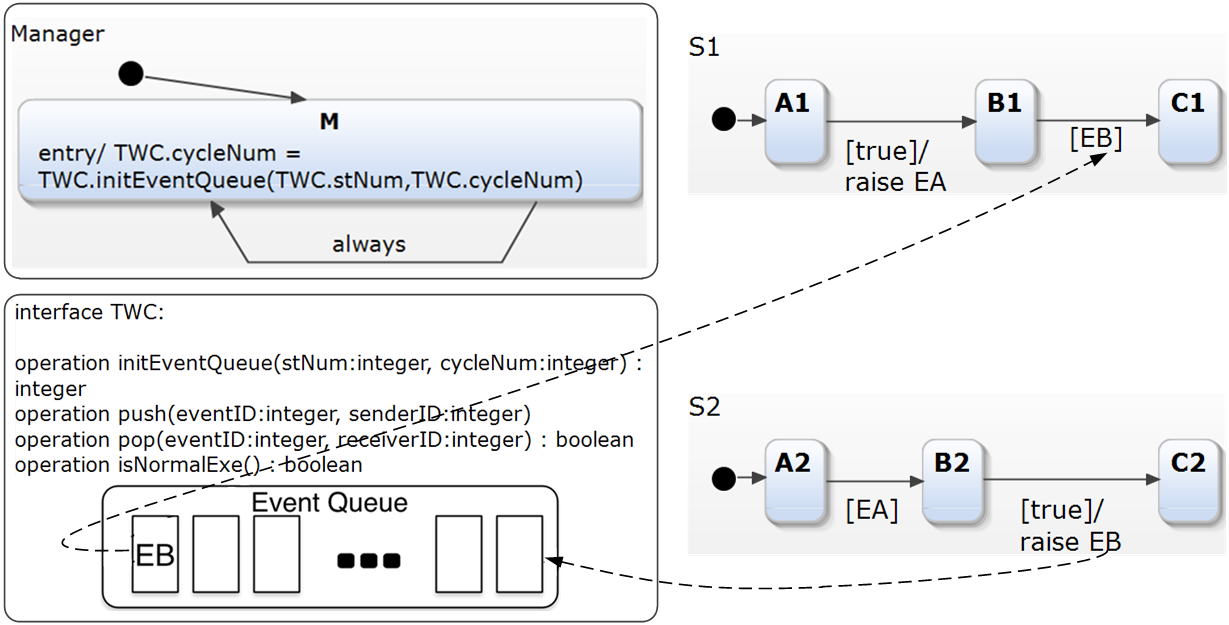}
%	\includegraphics[width = 0.99\textwidth]{fig/TWCpattern}
	%	\vspace{-1em}
	\caption{Two-Way Communication Model Pattern}
	\label{fig:TWCpattern}
\end{figure}

We define a Java class $\mathtt{TWC}$ shown in Fig.~\ref{fig:TWCcode}
to implement the $\mathtt{TWC}$ interface.
The event queue is implemented by two arrays $\mathtt{queuedEvents}$ and
$\mathtt{queuedEventsSender}$ to store raised events and their corresponding
senders.
The $\mathtt{TWC}$ class also includes
a boolean variable $\mathtt{normalExe}$ that indicates if the current
execution cycle is a \textit{normal cycle} or an \textit{logic cycle} and 
an integer variable $\mathtt{queuedEventNum}$ indicating the size of the
event queue.
The four functions declared in the $\mathtt{TWC}$ interface are implemented
as follows:
\begin{enumerate}
	\item $\mathtt{long \ initEventQueue(long \ stNum, long \ cycleNum)}$
	assigns the value of $\mathtt{normalExe}$ based on current execution cycle number $\mathtt{cycleNum}$,
	increases the execution cycle number by 1,
	and clears the event queue when the execution enters \textit{normal cycles}.
	
	\item $\mathtt{void \ push(long \ event, long \ sender)}$ pushes a raised
	event and the event raiser into the event queue. 
	
	\item $\mathtt{boolean \ pop(long \ event, long \ receiver)}$ checks if the 
	input $\mathtt{event}$ is valid. In \textit{normal cycles},
	events raised by higher priority stetecharts are
	valid; while in \textit{logic cycles}, events raised
	by lower priority stetecharts are valid.	
	
	\item $\mathtt{boolean \ isNormalExe()}$ checks if the current execution cycle
	is a \textit{normal cycle}.
\end{enumerate}
The $\mathtt{TWC}$ class only uses
basic Java data types, such as integer, boolean, and integer array,
and basic Java statements, such as assignment, \textit{if} statement, and \textit{while} statement.
The implementation principle of only using basic elements
decreases the difficulty of correctness proof in Section~\ref{subsec:TWCproof}.

\setlength{\intextsep}{5pt}
\setlength{\columnsep}{0pt}
\begin{figure}[ht]
	\centering
	%	\vspace{-2.5em}
	\includegraphics[width = 0.4\textwidth]{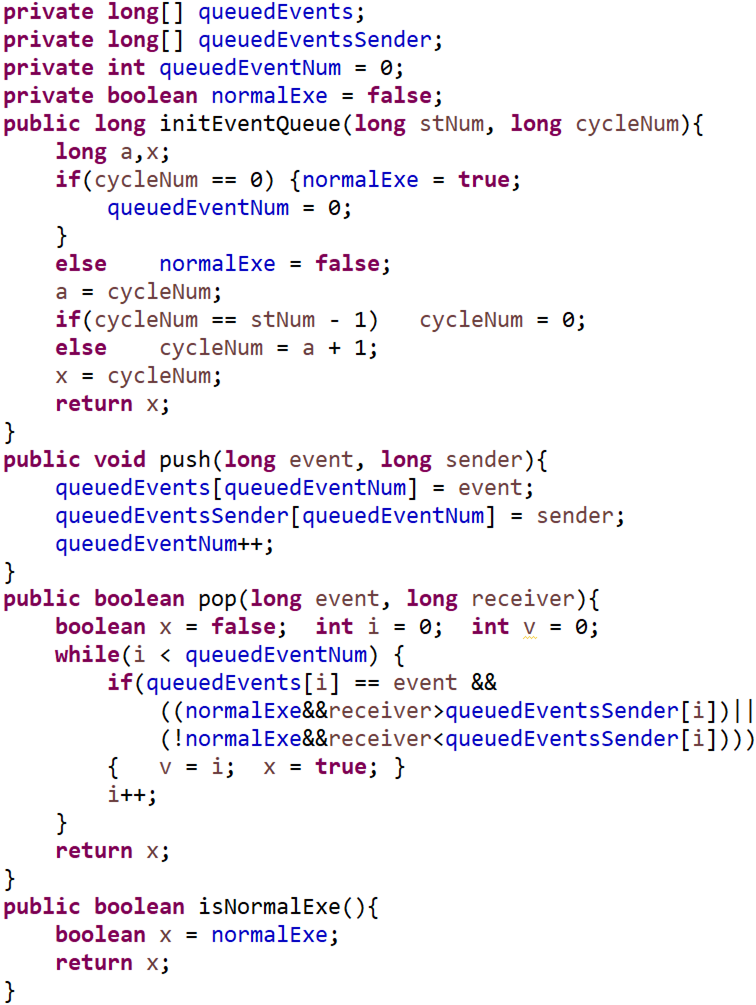}
%	\includegraphics[width = 0.9\textwidth]{fig/TWCcode}
	%	\vspace{-1em}
	\caption{Two-Way Communication Interface Implementation}
	\label{fig:TWCcode}
\end{figure}

We define Procedure~\ref{prc:applyTWC} to apply the two-way communication
model pattern in existing statechart models.

\begin{procedure}
\label{prc:applyTWC}
~
\begin{itemize}
	\item \textbf{Step 1:} add the interface $\mathtt{TWC}$ and the $\mathtt{Manager}$ statechart;
	
	\item \textbf{Step 2:} replace each event raise action by
	$\mathtt{TWC.push(TWC.\textit{eventID}, TWC.\textit{senderID})}$;
	
	\item \textbf{Step 3:} modify each transition's guard $G$ as follows:
	\begin{itemize}
		\item if $G$ does not contain any event, replace $G$ with
		$G \ \&\& \ \mathtt{TWC.isNormalExe()}$;
		
		\item if $G$ contains events, replace each event part
		in $G$ with corresponding expression
		$\mathtt{TWC.pop(TWC.\textit{eventID}, TWC.\textit{receiverID})}$.
	\end{itemize}

\end{itemize}	
\end{procedure}

\subsection{Design Model Pattern for Configurable Execution Order}
\label{subsec:EOM}

In statechart models with multiple statecharts, end users may
want to configure statechart execution orders 
based on their experiences and preferences.
We represent configurable execution order requirement by the statechart model shown in
Fig.~\ref{fig:EOMex1}.
If the execution order is $\mathtt{S1}$ $\prec$ $\mathtt{S2}$, the action
$x=x+1$ is executed first. While if
the execution order is $\mathtt{S2}$ $\prec$ $\mathtt{S1}$, the action
$y=y+1$ is executed first.
However, the statechart execution orders (represented by priorities)
are pre-assigned when building statechart
models and can not be changed without modifying existing statecharts.

\setlength{\intextsep}{5pt}
\setlength{\columnsep}{0pt}
\begin{figure}[ht]
	\centering
	%	\vspace{-2.5em}
	\includegraphics[width = 0.25\textwidth]{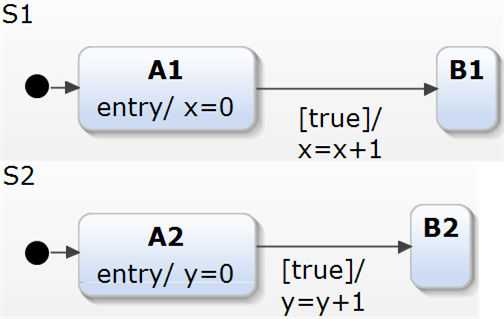}
%	\includegraphics[width = 0.6\textwidth]{fig/EOMex1}
	%	\vspace{-1em}
	\caption{Configurable Execution Order Feature}
	\label{fig:EOMex1}
\end{figure}

In this section, we design a model pattern to allow end users to
configure statechart execution orders without modifying existing
statechart models.
The design criteria is the same with the two-way communication
model pattern presented in Section~\ref{subsec:TWC}.
Based on the model pattern design criteria, our strategy to support
configurable execution order is to represent user specified execution
order in configure files, add \textit{logic execution cycles},
and apply token-based ordering
to achieve desired execution orders.

We also take \yakindu\ statecharts as an example to 
implement the configurable execution order model pattern.
In particular, the model pattern generates a token for each
\yakindu\ execution cycle based on a specified order.
For each \textit{normal execution cycle}, we add $n-1$
\textit{logic cycles}, where $n$ is the number of statecharts.
During an execution
cycle, only the statechart whose priority matches the generated token will execute
one step.

Similar to the two-way communication model pattern in Section~\ref{subsec:TWC},
the configurable execution order model pattern contains
a $\mathtt{Manager}$ statechart and
a interface $\mathtt{CEO}$, as shown in Fig.~\ref{fig:EOMpattern}.
In particular, the $\mathtt{Manager}$ statechart has the
highest priority in the model and updates the execution token
in each execution cycle.
The interface $\mathtt{CEO}$ declares two functions:
$\mathtt{updateExeInfo()}$ and $\mathtt{run()}$.

\setlength{\intextsep}{5pt}
\setlength{\columnsep}{0pt}
\begin{figure}[ht]
	\centering
	%	\vspace{-2.5em}
	\includegraphics[width = 0.49\textwidth]{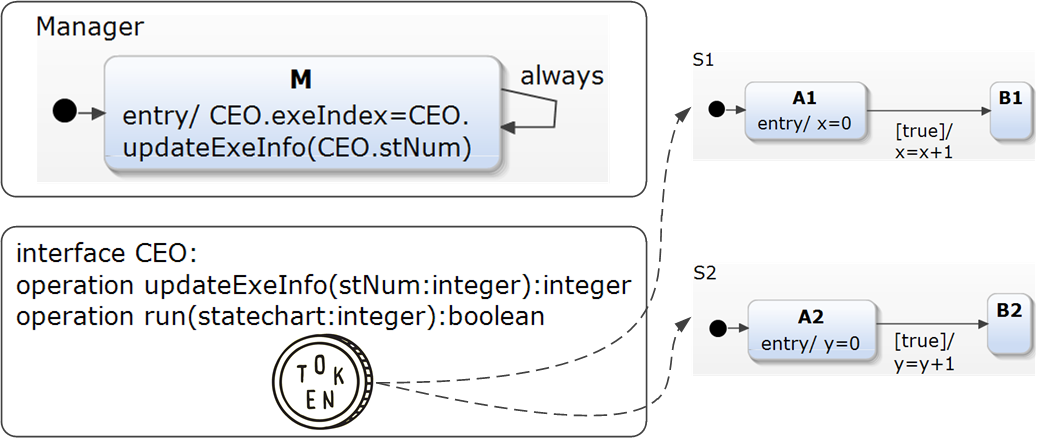}
%	\includegraphics[width = 0.99\textwidth]{fig/EOMpattern}
	%	\vspace{-1em}
	\caption{Configurable Execution Order Model Pattern}
	\label{fig:EOMpattern}
\end{figure}

We define a Java class $\mathtt{CEO}$ shown in Fig.~\ref{fig:EOMcode}
to implement the $\mathtt{CEO}$ interface.
The $\mathtt{CEO}$ class includes an integer variable $\mathtt{exeIndex}$
to represent the execution token and an array $\mathtt{exeOrders}$ to
store user specified execution orders, respectively.
The two functions declared in the $\mathtt{CEO}$ interface are implemented
as follows:
\begin{enumerate}
	\item $\mathtt{long \ updateExeInfo(long \ stNum)}$ updates the execution token,
	if the token is equal to the statechart number $\mathtt{stNum}$, the token
	is reset to be $1$, otherwise the token is increased by $1$;	
	
	\item $\mathtt{boolean \ run(long \ st)}$ checks if the input statechart $\mathtt{st}$
	matches the current execution token.
\end{enumerate}
The $\mathtt{CEO}$ class also only uses
basic Java data types and basic Java statements,
which decreases the difficulty of correctness proof in Section~\ref{subsec:EOMproof}.

\setlength{\intextsep}{5pt}
\setlength{\columnsep}{0pt}
\begin{figure}[ht]
	\centering
	\includegraphics[width = 0.33\textwidth]{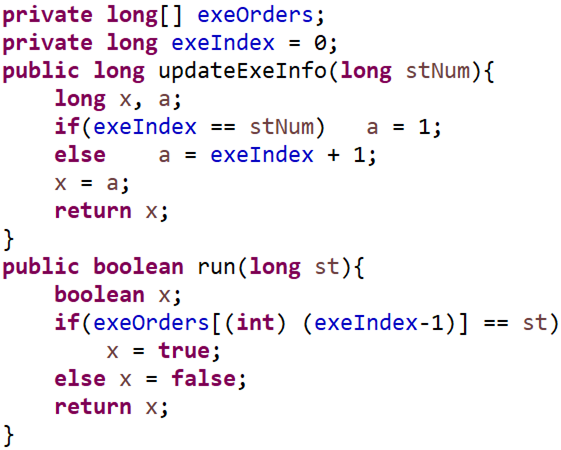}
%	\includegraphics[width = 0.7\textwidth]{fig/EOMcode}
	%	\vspace{-1em}
	\caption{Configurable Execution Order Interface Implementation}
	\label{fig:EOMcode}
\end{figure}

We define Procedure~\ref{prc:applyEOM} to apply the configurable execution order
model pattern in existing statechart models.

\begin{procedure}
	\label{prc:applyEOM}
	~
	\begin{itemize}
		\item \textbf{Step 1:} add the interface $\mathtt{CEO}$ and the $\mathtt{Manager}$ statechart;
		
		\item \textbf{Step 2:} replace each transition's guard $G$ by
		$G  \ \&\& \ \mathtt{CEO.run(CEO.\textit{statechartID})}$.		
	\end{itemize}	
\end{procedure}

In summary, we design two model patterns to support
two-way communication and configurable execution order in statecharts.
Both model patterns are implemented with basic statechart elements and external
Java codes. The approach has three advantages:
(1) it does not change statechart syntax and its underline execution semantics,
hence the two model patterns can be applied to all statechart models;
(2) it does not increase the difficulty of clinical validation and formal verification;
(3) existing work on modeling, validating, and verifying medical guidelines with statecharts
can be applied.
For safety-critical medical systems, the correctness of the designed model patterns
is crucial. We formally prove the two model patterns' correctness in Section~\ref{sec:proof}.

%% file: proof.tex
The model patterns designed in Section~\ref{sec:pattern} are implemented
with basic statechart elements and external Java code.
Our previous work~\cite{Guo2016ICCPS} have proved that the transformation
of basic statechart elements to timed automata is correct.
In this section, we prove the correctness of designed model patterns
in two steps:
(1) prove that the Java implementation is correct; 
and
(2) transform model patterns to timed autoamta to formally verify that
desired properties hold.
The strategy to prove the Java implementation is as follows:
(1) represent each Java function by a WHILE program~\cite{book536};
(2) construct a Hoare triple~\cite{book536} for each WHILE program;
and
(3) prove the correctness of the Hoare triples.

\subsection{Two-Way Communication Model Pattern}
\label{subsec:TWCproof}

We represent the Java functions $\mathtt{initEventQueue()}$,
$\mathtt{push()}$, $\mathtt{pop()}$, and $\mathtt{isNormalExe()}$
by WHILE programs Program~\ref{pro:initEventQueue}, Program~\ref{pro:push},
Program~\ref{pro:pop}, and Program~\ref{pro:isNormalExe}, respectively.
The variables used to implement the two-way communication model pattern 
are listed in Table~\ref{tab:TWC}.
We construct Hoare triples and prove their correctness in
Lemma~\ref{lm:initEventQueue}-\ref{lm:isNormalExe}.
The precondition and postcondition of each Hoare triple are the input and
output of the corresponding WHILE program based on the functionality.

\begin{table*}[h]
	\caption{Variables in Two-Way Communication Model Pattern Implementation} 
	\label{tab:TWC}
	\centering  
	\begin{tabular}{ | l | l | c | p{7cm} | } \hline
%	\begin{tabular}{ | p{2.7cm} | p{1.5cm} | c | p{5cm} | } \hline
		\textbf{Variable in \yakindu\ Java Code} & \textbf{Variable in WHILE Program} & \textbf{Variable Type} & \textbf{Variable Meaning} \\ \hline  
		$\mathtt{queuedEvents[]}$ & $E[]$ & int array & raised events \\ \hline
		$\mathtt{queuedEventsSender[]}$ & $S[]$ & int array & the event senders of corresponding element in $\mathtt{queuedEvents[]}$ \\ \hline
		$\mathtt{queuedEventNum}$ & $n$ & int & number of raised event, i.e., size of $E[]$ and $S[]$\\ \hline
		$\mathtt{normalExe}$ & $\mathtt{exe}$ & bool & a variable indicates if the current
		execution cycle is a normal cycle or an additional cycle \\ \hline
		$\mathtt{stNum}$ & $\mathtt{stNum}$ & int & statechart number of the given model \\ \hline
		$\mathtt{cycleNum}$ & $c$ & int & the number of execution cycle \\ \hline
		$\mathtt{event}$ & $e$ & int & event identity \\ \hline
		$\mathtt{sender}$ & $s$ & int & event sender \\ \hline
		$\mathtt{receiver}$ & $r$ & int & event receiver \\ \hline
		$v$ & $v$ & int & event $e$'s index in $E[]$ and $S[]$ \\ \hline
		$i$ & $i$ & int & iterator of integer arrays $E[]$ and $S[]$ \\ \hline
		$a$ & $a$ & int & temporary variable \\ \hline		
		$x$ & $x$ & int/bool & the return of functions \\ \hline
	\end{tabular} 
\end{table*}

\begin{program}
	\label{pro:initEventQueue}
	\begin{align*}
	& \mathtt{initEventQueue}\equiv	\\
	& \qquad \kwif c=0 \ \kwthen \ \mathtt{exe}:=\mathtt{true}; \ n:=0; \\
	& \qquad \kwelse \mathtt{exe}:=\mathtt{false}; \ \kwfi;\\
	& \qquad a:=c; \\
	& \qquad \kwif c=\mathtt{stNum}-1 \ \kwthen \ c:=0; \\
	& \qquad \kwelse c:=a+1; \ \kwfi;\\
	& \qquad x:=c
	\end{align*} 
\end{program}

The input of the $\mathtt{initEventQueue}$ program is the statechart number and
cycle number. The two-way communication requires as least two statecharts, hence
$\mathtt{stNum}>1$. As the two-way communication model pattern adds $\mathtt{stNum}-1$
\textit{logic cycle} and treats the \textit{normal cycle} as $0$,
hence $c < \mathtt{stNum}$.
Therefore, the precondition of the $\mathtt{initEventQueue}$ program
is $c < \mathtt{stNum} \land \mathtt{stNum}>1$.
According to the functionality, the $\mathtt{initEventQueue}$ program has two possible
outputs:
(1) if the current execution cycle is a \textit{normal cycle}, then program assigns the
number of next cycle to be $1$ and clears the event queue, i.e.,
$x=1 \land \mathtt{exe}=\mathtt{true} \land n=0$;
and
(2) if the current execution cycle is a \textit{logic cycle}, then program increases
the number of next cycle by $1$; if the current execution cycle is the last \textit{logic cycle},
then program assigns the number of next cycle to be $0$, i.e.,
$(x=0 \lor x=a+1) \land \mathtt{exe}=\mathtt{false}$.
Hence, the the postcondition of the $\mathtt{initEventQueue}$ program
is $(x=1 \land \mathtt{exe}=\mathtt{true} \land n=0) \lor ((x=0 \lor x=a+1) \land \mathtt{exe}=\mathtt{false})$.
We prove the correctness of the $\mathtt{initEventQueue}$ program in Lemma~\ref{lm:initEventQueue}.

\begin{lemma}
	\label{lm:initEventQueue}
	$\triple{c < \mathtt{stNum} \land \mathtt{stNum}>1}{\mathtt{initEventQueue}}{(x=1 \land \mathtt{exe}=\mathtt{true} \land n=0) \lor ((x=0 \lor x=a+1) \land \mathtt{exe}=\mathtt{false})}$
\end{lemma}
\begin{proof}
	According to the composition rule~\cite{book536}, the proof of the lemma is
	equivalent to prove the following two triples:
	\begin{align}
	\label{eq:initS1}
	\triple{P}{S_1;a:=c}{R}	
	\end{align}	
	\begin{align}
	\label{eq:initS2}
	\triple{R}{S_2;x:=c}{Q}	
	\end{align}
	where $P \equiv c < \mathtt{stNum} \land \mathtt{stNum}>1$,
	$Q \equiv (x=1 \land \mathtt{exe}=\mathtt{true} \land n=0) \lor ((x=0 \lor x=c+1) \land \mathtt{exe}=\mathtt{false})$,
	$R \equiv (a=c \land a=0 \land \mathtt{exe}=\mathtt{true} \land n=0 \land \mathtt{stNum}>1) \ \lor \ (a=c \land a < \mathtt{stNum} \land \mathtt{exe}=\mathtt{false} \land \mathtt{stNum}>1)$,
	\begin{align*}
	S_1 \equiv
	& \quad \kwif c=0 \ \kwthen \ \mathtt{exe}:=\mathtt{true}; \ n:=0; \\
	& \quad \kwelse \mathtt{exe}:=\mathtt{false}; \ \kwfi
	\end{align*}
	and
	\begin{align*}
	S_2 \equiv
	& \quad \kwif c=\mathtt{stNum}-1 \ \kwthen \ c:=0; \\
	& \quad \kwelse c:=a+1; \ \kwfi
	\end{align*}
	
	The detailed proof of triple~\eqref{eq:initS1}
	and triple~\eqref{eq:initS2} is given in Appendix~\ref{subsec:initEventQueueProof}.
\end{proof}

\begin{program}
	\label{pro:push}
	\begin{align*}
	\mathtt{push}\equiv
	& \quad E[n]:=e;\ S[n]:=s; \ n:=n+1
	\end{align*}
\end{program}

The input of the $\mathtt{push}$ program is an event $e$,
its sender $s$, the event array $E[]$, and the sender array $S[]$.
Suppose the number of raised \textit{events} before pushing event $e$
is $N$, then $n=N \land N \ge 0$.
The event $e$ and its sender $s$
are represented with positive integers, hence $e>0 \land s>0$.
The precondition of the $\mathtt{push}$ program is $n=N \land N \ge 0 \land e>0 \land s>0$.
The $\mathtt{push}$ program pushes an event and its sender into the queues
and increase the number of raised \textit{events} by 1.
The event array $E[]$ and the sender array $S[]$ have
the same size $n$, i.e., the number of raised \textit{events}.
As the index number of arrays implementing queues starts from $0$,
the postcondition of the $\mathtt{push}$ program is $E[n-1]=e \land S[n-1]=s \land n=N+1$. 
We prove the correctness of the $\mathtt{push}$ program in Lemma~\ref{lm:push}.

\begin{lemma}
	\label{lm:push}
	$\triple{n=N \land N \ge 0 \land e>0 \land s>0}{\mathtt{push}}{E[n-1]=e \land S[n-1]=s \land n=N+1}$
\end{lemma}
\begin{proof}
	\begin{align*}
	& \quad E[n-1]=e \land S[n-1]=s \ \land \\
	& \quad n=N+1 [n:=n+1][S[n]:=s][E[n]:=e]\\
	& \equiv E[n]=e \land S[n]=s \land n=N [S[n]:=s][E[n]:=e]\\
	& \equiv E[n]=e \land n=N [E[n]:=e]\\
	& \equiv n=N
	\end{align*}
	
	As $n=N \land N \ge 0 \land e>0 \land s>0 \rightarrow n=N$, hence the lemma is correct.
	%\qed
\end{proof}

\begin{program}
	\label{pro:pop}
	\begin{align*}
	\mathtt{pop} \equiv
	& \quad x:=\mathtt{false}; \ i:=0; \ v:=0; \\
	& \quad \kwwhile i < n \ \kwdo \\
	& \quad \qquad \kwif E[i] = e \land ((\mathtt{exe}=\mathtt{true} \land r > S[i]) \ \lor \\
	& \quad \qquad \qquad (\mathtt{exe}=\mathtt{false} \land r < S[i])) \ \kwthen \\
	& \quad \qquad \qquad v:=i; \ x := \mathtt{true}; \ \kwfi; \\
	& \quad \qquad i:=i+1; \\
	& \quad \kwod
	\end{align*} 
\end{program}

The input of the $\mathtt{pop}$ program is an event and
its receiver, hence, the precondition of the $\mathtt{pop}$ program is $e>0 \land r>0$.
The $\mathtt{pop}$ program checks if an event is valid to a receiver during current execution cycle.
In a \textit{normal cycle}, events raised by higher priority stetecharts are
valid, i.e., $x=\mathtt{true} \land E[v]=e \land 0 \le i \le n \land \mathtt{exe}=\mathtt{true} \land r>S[v]$, where $v$ is event $e$'s index in event array $E$.
Note that a \textit{event} $e$ and its sender $s$ have the same index
in the event array $E[]$ and the sender array $S[]$, respectively.
While in an \textit{logic cycle}, events raised by lower priority stetecharts are valid,
i.e., $x=\mathtt{true} \land E[v]=e \land 0 \le i \le n \land \mathtt{exe}=\mathtt{false} \land r<S[v]$.
Note that the higher priority a statechart has, the lower its identity is.
If the input event is not raised, the $\mathtt{pop}$ program returns $\mathtt{false}$,
i.e., $x=\mathtt{false}$.
Therefore, the postcondition of the $\mathtt{pop}$ program is
$x=\mathtt{false} \lor (x=\mathtt{true} \land E[v]=e \land ((\mathtt{exe}=\mathtt{true} \land r>S[v]) \lor (\mathtt{exe}=\mathtt{false} \land r<S[v])))$.
We prove the correctness of the $\mathtt{pop}$ program in Lemma~\ref{lm:pop}.

\begin{lemma}
	\label{lm:pop}
	$\triple{e>0 \land r>0}{\mathtt{pop}}{x=\mathtt{false} \lor (x=\mathtt{true} \land E[v]=e \land ((\mathtt{exe}=\mathtt{true} \land r>S[v]) \lor (\mathtt{exe}=\mathtt{false} \land r<S[v])))}$
\end{lemma}
\begin{proof}
	The loop invariant and bound function of the $\kwwhile$ loop
	in the $\mathtt{pop}$ program are
	$\mathtt{\textbf{inv}} \equiv 0 \le i \le n \land (x=\mathtt{false} \lor (x=\mathtt{true} \land E[v]=e \land ((\mathtt{exe}=\mathtt{true} \land r>S[v]) \lor (\mathtt{exe}=\mathtt{false} \land r<S[v]))))$
	and
	$\mathtt{\textbf{bd}} \equiv n-i$, respectively.

	We use $P$ and $Q$ to denote the precondition and postcondition
	of the triple, i.e.,
	$P \equiv e>0 \land r>0$ and
	$Q \equiv x=\mathtt{false} \lor (x=\mathtt{true} \land E[v]=e \land ((\mathtt{exe}=\mathtt{true} \land r>S[v]) \lor (\mathtt{exe}=\mathtt{false} \land r<S[v])))$.
	We prove the lemma in the following five steps.
	
	~\newline
	\textbf{Step 1:} prove that the initialization establishes
	the loop invariant, i.e., $\triple{P}{x:=\mathtt{false}; i:=0; v:=0}{\mathtt{inv}}$.
	
%	The proof of triple $\triple{P}{x:=\mathtt{false}; i:=0; v:=0}{\mathtt{inv}}$
%	is as follows.	
%	\begin{align*}
%	& \quad \ \mathtt{inv}[v:=0][i:=0][x:=\mathtt{false}]\\	
%	& \equiv 0 \le i \le n \land (x=\mathtt{false} \lor (x=\mathtt{true} \ \land \\
%	& \quad \ E[v]=e \ \land \ ((\mathtt{exe}=\mathtt{true} \land r>S[v]) \ \lor \\
%	& \quad \ (\mathtt{exe}=\mathtt{false} \land r<S[v])))) [v:=0][i:=0][x:=\mathtt{false}] \\	
%	& \equiv 0 \le i \le n \land (x=\mathtt{false} \lor (x=\mathtt{true} \ \land \\
%	& \quad \ E[0]=e \ \land \ ((\mathtt{exe}=\mathtt{true} \land r>S[0]) \ \lor \\
%	& \quad \ (\mathtt{exe}=\mathtt{false} \land r<S[0])))) [i:=0][x:=\mathtt{false}] \\	
%	& \equiv \mathtt{true} \land (x=\mathtt{false} \lor (x=\mathtt{true} \ \land \\
%	& \quad \ E[0]=e \ \land \ ((\mathtt{exe}=\mathtt{true} \land r>S[0]) \ \lor \\
%	& \quad \ (\mathtt{exe}=\mathtt{false} \land r<S[0])))) [x:=\mathtt{false}] \\	
%	& \equiv \mathtt{true} \land (\mathtt{true} \lor (\mathtt{false} \ \land \\
%	& \quad \ E[0]=e \ \land \ ((\mathtt{exe}=\mathtt{true} \land r>S[0]) \ \lor \\
%	& \quad \ (\mathtt{exe}=\mathtt{false} \land r<S[0])))) \\
%	& \equiv \mathtt{true}		
%	\end{align*}
%	As $P \rightarrow \mathtt{true}$, hence the triple
%	$\triple{P}{x:=\mathtt{false}; i:=0; v:=0}{\mathtt{inv}}$
%	is correct.

	~\newline
	\textbf{Step 2:} prove that the loop body does not change
	the loop invariant, i.e., $\triple{\mathtt{inv} \land i < n}{S}{\mathtt{inv}}$,
	where
	\begin{align*}
	S \equiv	
	& \quad \kwif E[i] = e \land (\mathtt{exe}=\mathtt{true} \land r > S[i]) \ \lor \\
	& \quad \qquad (\mathtt{exe}=\mathtt{false} \land r < S[i]) \ \kwthen \\
	& \quad \qquad v:=i; \ x := \mathtt{true}; \ \kwfi; \\
	& \quad i:=i+1
	\end{align*}

	~\newline
	\textbf{Step 3:} prove that the bound function decreases after each iteration,
	i.e., $\triple{\mathtt{inv} \land i < n \land \mathtt{bd} = z}{S}{\mathtt{bd} < z}$,
	where $z$ is an integer.

	~\newline	
	\textbf{Step 4:} prove that the loop invariant implies the bound
	function is non-negative, i.e., $\mathtt{inv} \rightarrow \mathtt{bd} \ge 0$.
	
%	As $\mathtt{bd} \ge 0 \equiv n \ge i$, hence $\mathtt{inv} \rightarrow \mathtt{bd} \ge 0$.	
	
	~\newline
	\textbf{Step 5:} prove that the postcondition holds when the
	$\kwwhile$ loop terminates, i.e., $\mathtt{inv} \land \neg (i < n) \rightarrow Q$.
	
%	\begin{align*}
%	& \quad \ \mathtt{inv} \land \neg (i < n) \\	
%	& \equiv 0 \le i \le n \land (x=\mathtt{false} \lor (x=\mathtt{true} \land \\
%	& \quad \ E[v]=e \ \land \ ((\mathtt{exe}=\mathtt{true} \land r>S[v]) \ \lor \\
%	& \quad \ (\mathtt{exe}=\mathtt{false} \land r<S[v])))) \ \land i \ge n \\
%	& \equiv i=n \land (x=\mathtt{false} \lor (x=\mathtt{true} \land E[v]=e \land \\
%	& \quad \ ((\mathtt{exe}=\mathtt{true} \land r>S[v]) \lor (\mathtt{exe}=\mathtt{false} \land r<S[v])))) \\
%	& \rightarrow Q
%	\end{align*}

	~\newline
	The detailed proof of each step is given in Appendix~\ref{subsec:popProof}.
%	
	%\qed
\end{proof}

\begin{program}
	\label{pro:isNormalExe}
	\begin{align*}
	\mathtt{isNormalExe}\equiv \quad x:=\mathtt{exe}
	\end{align*}
\end{program}

The $\mathtt{isNormalExe}$ program does not have input, hence,
the precondition is $\mathtt{true}$.
As the $\mathtt{isNormalExe}$ program checks if the current execution cycle
is a \textit{normal cycle}, the postcondition is $x=\mathtt{exe}$. 
We prove the correctness of the $\mathtt{isNormalExe}$ program in Lemma~\ref{lm:isNormalExe}.

\begin{lemma}
	\label{lm:isNormalExe}
	$\triple{\mathtt{true}}{\mathtt{isNormalExe}}{x=\mathtt{exe}}$
\end{lemma}
\begin{proof}
	\begin{align*}
	& \quad x=\mathtt{exe}[x:=\mathtt{exe}] \\
	& \equiv \mathtt{exe}=\mathtt{exe} \\
	& \equiv \mathtt{true}
	\end{align*}
	%\qed
\end{proof}

After proving that the Java implementation is correct,
we prove the correctness of the two-way communication model pattern in
Theorem~\ref{thm:TWC}.

\begin{theorem}
	\label{thm:TWC}
	The two-way communication model pattern is correct, i.e.,	
	two statecharts can send \textit{events} to each other.
\end{theorem}
\begin{proof}
	For the statechart model shown in Fig.~\ref{fig:TWCex1},	
	the communication from statechart $\mathtt{S1}$ to $\mathtt{S2}$
	means that statechart $\mathtt{S2}$ can receive \textit{event} $\mathtt{EA}$ sent
	by statechart $\mathtt{S1}$ and reach state $\mathtt{B2}$. Similarly,
	the communication from statechart $\mathtt{S2}$ to $\mathtt{S1}$
	is that statechart $\mathtt{S1}$ can receive \textit{event} $\mathtt{EB}$ sent
	by statechart $\mathtt{S2}$ and reach state $\mathtt{C1}$.
	Hence, to prove the correctness of the two-way communication model pattern
	is equivalent to prove that both state $\mathtt{C1}$ and state $\mathtt{B2}$
	are reachable.
	
	We have prove the correctness of transformation of statechart elements
	in~\cite{Guo2016ICCPS} and correctness of the Java implementation
	in Lemma~\ref{lm:initEventQueue}-\ref{lm:isNormalExe}.
	Hence, we can transform the statechart model in Fig.~\ref{fig:TWCex1} with two-way
	communication model pattern to timed automata to verify the properties.
	The two state reachability properties can be formally verified with
	TCTL (timed computation tree logic)~\cite{behrmann2004tutorial} formulas $E<> \ \mathtt{S1.C1}$
	and $E<> \ \mathtt{S2.B2}$, respectively.
	The verification results show that both properties are satisfied.
	Therefore, the two-way communication model pattern is correct.
	%\qed
\end{proof}

\subsection{Configurable Execution Order Model Pattern}
\label{subsec:EOMproof}

We represent the Java functions $\mathtt{updateExeInfo()}$, and
$\mathtt{run()}$
by WHILE programs Program~\ref{pro:updateExeInfo} and Program~\ref{pro:run}, respectively.
The variables used to implement the configurable execution order model pattern 
are listed in Table~\ref{tab:EOM}.
We construct Hoare triples and prove their correctness in
Lemma~\ref{lm:updateExeInfo}-\ref{lm:run}.

\begin{table*}[h]
	\caption{Variables in Configurable Execution Order Model Pattern Implementation} 
	\label{tab:EOM}
	\centering 
	\begin{tabular}{ | l | l | c | l | } \hline 
%	\begin{tabular}{ | p{2.7cm} | p{1.5cm} | c | p{5cm} | } \hline
		\textbf{Variable in \yakindu\ Java Code} & \textbf{Variable in WHILE Program} & \textbf{Variable Type} & \textbf{Variable Meaning} \\ \hline  
		$\mathtt{exeOrders[]}$ & $O[]$ & int array & statechart execution orders \\ \hline
		$\mathtt{exeIndex}$ & $t$ & int & execution token \\ \hline
		$\mathtt{stNum}$ & $\mathtt{stNum}$ & int & statechart number of the given model \\ \hline		
		$\mathtt{st}$ & $\mathtt{st}$ & int & statechart identity \\ \hline	
		$a$ & $a$ & int & temporary variable \\ \hline		
		$x$ & $x$ & int/bool & the return of functions \\ \hline
	\end{tabular} 
\end{table*}

\begin{program}
	\label{pro:updateExeInfo}
	\begin{align*}
	\mathtt{updateExeInfo}\equiv
	& \quad \kwif t = \mathtt{stNum} \ \kwthen \ a := 1; \\
	& \quad \kwelse a := t + 1; \ \kwfi;\\
	& \quad x := a
	\end{align*} 
\end{program}

The input of the $\mathtt{updateExeInfo}$ program is the statechart number
and the execution token, hence, the
precondition of the $\mathtt{updateExeInfo}$ program is $t>0 \land \mathtt{stNum}>0$.
The $\mathtt{updateExeInfo}$ program updates the execution token.
If the current execution token is equal to the statechart number, then
the program assigns the token to be $1$, i.e.,
$t=\mathtt{stNum} \land x=1$; otherwise, the program
increases the token by $1$, i.e., $t \ne \mathtt{stNum} \land x=t+1$.
Hence, the postcondition of the $\mathtt{updateExeInfo}$ program is
$(t=\mathtt{stNum} \land x=1) \lor (t \ne \mathtt{stNum} \land x=t+1)$.
We prove the correctness of the $\mathtt{updateExeInfo}$ program in Lemma~\ref{lm:updateExeInfo}.

\begin{lemma}
	\label{lm:updateExeInfo}
	$\triple{t>0 \land \mathtt{stNum}>0}{\mathtt{updateExeInfo}}{(t=\mathtt{stNum} \land x=1) \lor (t \ne \mathtt{stNum} \land x=t+1)}$
\end{lemma}
\begin{proof}
	According to the composition rule~\cite{book536}, the proof of the lemma is
	equivalent to prove the following two triples:
	\begin{align}
	\label{eq:updateExeInfoS1}
	\triple{P}{S}{R}	
	\end{align}
	\begin{align}
	\label{eq:updateExeInfoS2}
	\triple{R}{x:=a}{Q}	
	\end{align}
	where $P \equiv t>0 \land \mathtt{stNum}>0$,
	$Q \equiv (t=\mathtt{stNum} \land x=1) \lor (t \ne \mathtt{stNum} \land x=t+1)$,
	$R \equiv (t=\mathtt{stNum} \land a=1) \lor (t \ne \mathtt{stNum} \land a=t+1)$,
	and
	\begin{align*}
	S \equiv
	& \quad \kwif t = \mathtt{stNum} \ \kwthen \ a := 1; \\
	& \quad \kwelse a := t + 1; \ \kwfi
	\end{align*}
	
	The detailed proof of triple~\eqref{eq:updateExeInfoS1}
	and triple~\eqref{eq:updateExeInfoS2} is given in Appendix~\ref{subsec:updateExeInfoProof}.
\end{proof}

\begin{program}
	\label{pro:run}
	\begin{align*}
	\mathtt{run}\equiv
	& \quad \kwif O[t-1] = \mathtt{st} \ \kwthen \ x := \mathtt{true}; \\
	& \quad \kwelse x := \mathtt{false}; \ \kwfi
	\end{align*}
\end{program}

The input of the $\mathtt{run}$ program is a statechart's identity and
an execution token.
As we define that the minimal statechart identity and minimal generated execution token are both $1$, hence
the precondition of the $\mathtt{run}$ program is $t>0 \land \mathtt{st}>0$.
The $\mathtt{run}$ program checks if a statechart is valid to execute, i.e.,
the statechart's identity matches the current execution token.
The array $O[]$ stores statechart identities sorted in their execution orders.
The index number of the array $O[]$ starts from $0$.
Hence, the postcondition of the $\mathtt{run}$ program is
$(x=\mathtt{true} \land O[t-1] = \mathtt{st}) \lor (x=\mathtt{false} \land O[t-1] \ne \mathtt{st})$.
We prove the correctness of the $\mathtt{run}$ program in Lemma~\ref{lm:run}.

\begin{lemma}
	\label{lm:run}
	$\triple{t>0 \land \mathtt{st} > 0}{\mathtt{run}}{(x=\mathtt{true} \land O[t-1] = \mathtt{st}) \lor (x=\mathtt{false} \land O[t-1] \ne \mathtt{st})}$
\end{lemma}
\begin{proof}
	The proof of the lemma is equivalent to prove the following two triples:
	\begin{align}
	\label{eq:runS1}
	\triple{P \land O[t-1] = \mathtt{st}}{x := \mathtt{true}}{Q}	
	\end{align}	
	\begin{align}
	\label{eq:runS2}
	\triple{P \land O[t-1] \ne \mathtt{st}}{x := \mathtt{false}}{Q}	
	\end{align}
	where $P \equiv t>0 \land \mathtt{st} > 0$ and 
	$Q \equiv (x=\mathtt{true} \land O[t-1] = \mathtt{st}) \lor (x=\mathtt{false} \land O[t-1] \ne \mathtt{st})$.
	
	The detailed proof of triple~\eqref{eq:runS1}
	and triple~\eqref{eq:runS2} is given in Appendix~\ref{subsec:runProof}.
%		
%	The proof of triple~\eqref{eq:runS1} is as follows.
%	\begin{align*}
%	& \quad Q [x := \mathtt{true}] \\
%	& \equiv (x=\mathtt{true} \land O[t-1] = \mathtt{st}) \ \lor \\
%	& \quad \ (x=\mathtt{false} \land O[t-1] \ne \mathtt{st}) [x := \mathtt{true}] \\	
%	& \equiv (\mathtt{true}=\mathtt{true} \land O[t-1] = \mathtt{st}) \ \lor \\
%	& \quad \ (\mathtt{true}=\mathtt{false} \land O[t-1] \ne \mathtt{st}) \\	
%	& \equiv O[t-1] = \mathtt{st}
%	\end{align*}
%	As $P \land O[t-1] = \mathtt{st} \rightarrow O[t-1] = \mathtt{st}$,
%	hence the triple~\eqref{eq:runS1} is correct.
%	
%	Similarly, the proof of triple~\eqref{eq:runS2} is as follows.
%	\begin{align*}
%	& \quad Q [x := \mathtt{false}] \\
%	& \equiv (x=\mathtt{true} \land O[t-1] = \mathtt{st}) \ \lor \\
%	& \quad \ (x=\mathtt{false} \land O[t-1] \ne \mathtt{st}) [x := \mathtt{false}] \\	
%	& \equiv (\mathtt{false}=\mathtt{true} \land O[t-1] = \mathtt{st}) \ \lor \\
%	& \quad \ (\mathtt{false}=\mathtt{false} \land O[t-1] \ne \mathtt{st}) \\	
%	& \equiv O[t-1] \ne \mathtt{st}
%	\end{align*}
%	As $P \land O[t-1] \ne \mathtt{st} \rightarrow O[t-1] \ne \mathtt{st}$,
%	hence the triple~\eqref{eq:runS2} is correct.
%	
%	Therefore, the lemma is correct.
	%\qed
\end{proof}

After proving that the Java implementation is correct,
we prove the correctness of the configurable execution order model pattern in
Theorem~\ref{thm:EOM}.

\begin{theorem}
	\label{thm:EOM}
	The configurable execution order model pattern is correct, i.e.,
	the statechart execution order is the same with the configuration.
\end{theorem}
\begin{proof}
	For the statechart model shown in Fig.~\ref{fig:EOMex1},
	if the execution order is $\mathtt{S1}$ $\prec$ $\mathtt{S2}$, the action
	$x=x+1$ is executed before $y=y+1$, i.e., $x \ge y$. If
	the execution order is $\mathtt{S2}$ $\prec$ $\mathtt{S1}$, the action
	$y=y+1$ is executed before $x=x+1$, i.e., $y \ge x$.
	Hence, to prove the correctness of the configurable execution order model pattern
	is equivalent to prove the above two properties under
	corresponding execution order configurations.
		
	We have prove the correctness of transformation of statechart elements
	in~\cite{Guo2016ICCPS} and correctness of the Java implementation
	in Lemma~\ref{lm:updateExeInfo} and Lemma~\ref{lm:run}.
	Hence, we can transform the statechart model in Fig.~\ref{fig:EOMex1} with configurable execution order
	model pattern to timed automata to verify the properties.
	The two properties can be formally verified with
	TCTL (timed computation tree logic)~\cite{behrmann2004tutorial} formulas $A[~] \ x \ge y$
	and $A[~] \ y \ge x$, respectively.
	We specify statechart execution orders to be $\mathtt{S1}$ $\prec$ $\mathtt{S2}$
	and $\mathtt{S2}$ $\prec$ $\mathtt{S1}$.
	The verification results show that the properties hold under
	corresponding execution order configurations.
	Therefore, the configurable execution order model pattern is correct.
	%\qed
\end{proof}

%% file: exp.tex
In this section, we use the simplified airway laser surgery scenario presented
in Example~\ref{ex:laser} to demonstrate how the designed model patterns can facilitate developing executable medical guideline models and their impact on validation and verification of medical safety properties.

The simplified airway laser surgery scenario has two safety properties, i.e.,
\textbf{P1}: the laser and the ventilator can not be activated at the
same time;
and 
\textbf{P2}: the patient's $\mathtt{SpO}$ level can not be smaller
than $95\%$.
We run simulation of the simplified airway laser surgery statechart model
in Fig.~\ref{fig:laser} through \yakindu. The simulation results show
that the model reaches an unsafe state that
both $\mathtt{Laser}$ and $\mathtt{Ventilator}$ are on.
In addition, we use the \toolname\ tool~\cite{Guo2016ICCPS} to automatically transform the statechart model in Fig.~\ref{fig:laser}
to timed automata to formally verify the two safety properties.
The safety properties \textbf{P1} and \textbf{P2}
are verified in \uppaal\ by TCTL~\cite{behrmann2004tutorial} formulas
$A[~] \ ! (\mathtt{Laser.On}  \ \&\& \ \mathtt{Ventilator.On})$
and $A[~] \ \mathtt{SpO} >= 95$, respectively.
The verification results show that the safety property \textbf{P1}
indeed fails.
We trace back the execution path that fails \textbf{P1} to the statechart
model and find that the $\mathtt{deactivateLaser}$ \textit{event} sent by
$\mathtt{ventilator}$ can not be received by the $\mathtt{laser}$.
The reason is that \yakindu\ statecharts' priority-based execution semantics
cause that lower priority statecharts can not
send \textit{events} to higher priority statecharts.

We use Procedure~\ref{prc:applyTWC} and Procedure~\ref{prc:applyEOM}
to apply the two-way communication model pattern and the configurable
execution order model pattern to the
statechart model in Fig.~\ref{fig:laser}. The simplified airway laser
surgery statechart with model patterns is shown in Fig.~\ref{fig:laserPattern}.
We run simulation of the statechart model
in Fig.~\ref{fig:laserPattern} through \yakindu. The simulation results show
that the model does not reach any unsafe states.
We also transform the statechart model in Fig.~\ref{fig:laserPattern} to timed
automata, as shown in Fig.~\ref{fig:laserPatternU}.
The verification results also indicate that both safe properties \textbf{P1}
and \textbf{P2} are satisfied.

\setlength{\intextsep}{5pt}
\setlength{\columnsep}{0pt}
\begin{figure}[ht]
	\centering
	\includegraphics[width = 0.49\textwidth]{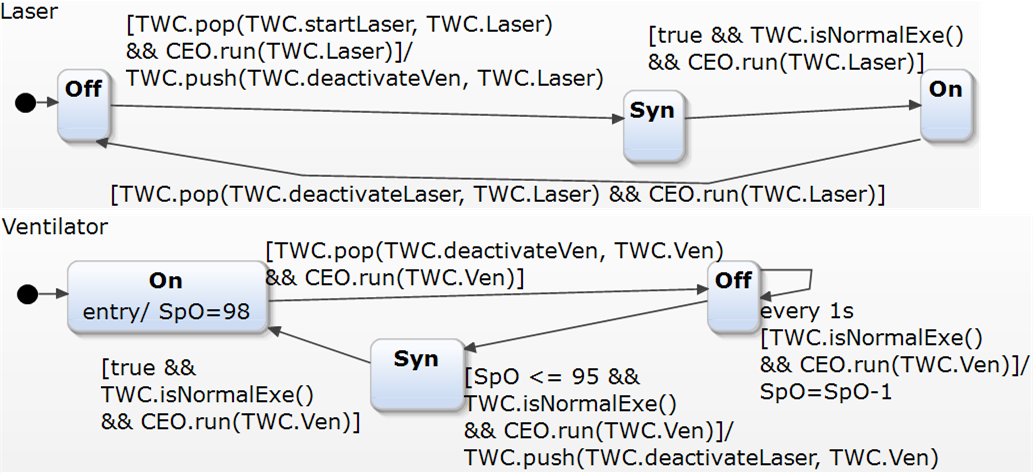}
%	\includegraphics[width = 0.99\textwidth]{fig/laserPattern}
	%	\vspace{-1em}
	\caption{Simplified Airway Laser Surgery Statecharts with Model Patterns}
	\label{fig:laserPattern}
\end{figure}

\setlength{\intextsep}{5pt}
\setlength{\columnsep}{0pt}
\begin{figure}[ht]
	\centering
	\includegraphics[width = 0.49\textwidth]{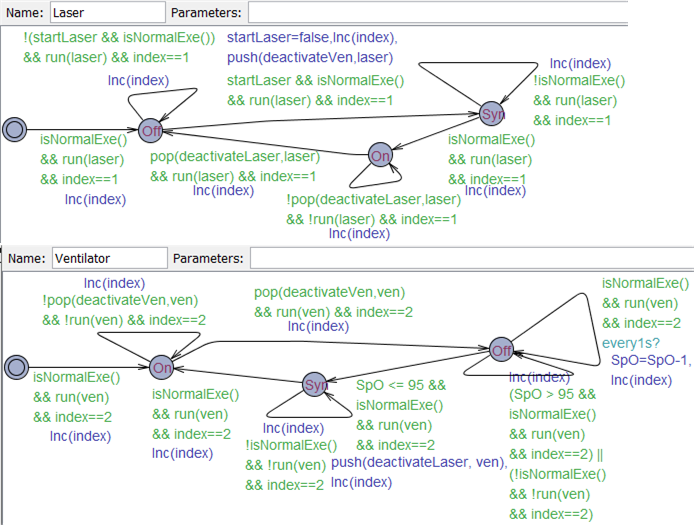}
%	\includegraphics[width = 0.99\textwidth]{fig/laserPatternU}
	%		\vspace{-1em}
	\caption{Simplified Airway Laser Surgery Timed Automata Model with Model Patterns}
	\label{fig:laserPatternU}
\end{figure}

The default statechart execution order of the statechart model in Fig.~\ref{fig:laserPattern}
is $\mathtt{Laser}$ $\prec$ $\mathtt{Ventilator}$.
Some medical professionals may prefer the execution order to be
$\mathtt{Ventilator}$ $\prec$ $\mathtt{Laser}$.
We specify the statechart execution order to be $\mathtt{Ventilator}$ $\prec$ $\mathtt{Laser}$
and validate/verify the two safety properties.
Both simulation and verification results show that \textbf{P1} and \textbf{P2}
are still satisfied.

\modify{
We also use the simplified airway laser surgery to illustrate
the correctness of designed model patterns when they are accompanied
by medical guideline models.
The correctness of the two-way communication model pattern is that
two statecharts can send events to each other.
For the simplified airway laser surgery model shown in
Fig.~\ref{fig:laserPattern},	
the communication from statechart $\mathtt{Laser}$ to $\mathtt{Ventilator}$
means that statechart $\mathtt{Ventilator}$ can receive \textit{event} $\mathtt{deactivateVen}$ sent
by statechart $\mathtt{Laser}$ and reach state $\mathtt{Ventilator.Off}$. Similarly,
the communication from statechart $\mathtt{Ventilator}$ to $\mathtt{Laser}$
is that statechart $\mathtt{Laser}$ can receive \textit{event} $\mathtt{deactivateLaser}$ sent
by statechart $\mathtt{Ventilator}$ and reach state $\mathtt{Laser.Off}$.
Hence, the correctness of two-way communication model pattern is that
both state $\mathtt{Ventilator.Off}$ and state $\mathtt{Laser.Off}$ are reachable.
We verify the two state reachability properties in the timed automata
model shown in Fig.~\ref{fig:laserPatternU} with TCTL formulas
$E<> \ \mathtt{Ventilator.Off}$ and $E<> \ \mathtt{Laser.Off}$, respectively.
The verification results show that both properties are satisfied,
i.e., the two-way communication model pattern is correct in
the simplified airway laser surgery model.

The correctness of the configurable execution order model pattern
is that the statechart execution order is the same with the configuration.
For the simplified airway laser surgery model shown in
Fig.~\ref{fig:laserPattern},
if the execution order is $\mathtt{Laser}$ $\prec$ $\mathtt{Ventilator}$,
the initial state $\mathtt{Laser.Off}$ is reached before the initial
state $\mathtt{Ventilator.On}$, which is represented by TCTL formula
$A[~] \ \mathtt{Ventilator.On} \ imply \ \mathtt{Laser.Off}$.
If the execution order is $\mathtt{Ventilator}$ $\prec$ $\mathtt{Laser}$,
the initial state $\mathtt{Ventilator.On}$ is reached before the initial
state $\mathtt{Laser.Off}$, which is represented by TCTL formula
$A[~] \ \mathtt{Laser.Off} \ imply \ \mathtt{Ventilator.On}$.
We specify statechart execution orders to be $\mathtt{Laser}$ $\prec$ $\mathtt{Ventilator}$
and $\mathtt{Ventilator}$ $\prec$ $\mathtt{Laser}$ and verify
corresponding properties in the timed automata model shown in
Fig.~\ref{fig:laserPatternU}, respectively.
The verification results show that the properties hold under
corresponding execution order configurations,
i.e., the configurable execution order model pattern is correct in
the simplified airway laser surgery model.
}

The case study demonstrates that
(1) the two-way communication and configurable execution order functionalities are
crucial in modeling medical guidelines;
and
(2) the designed model patterns can address the two-way communication and configurable execution order issues.

%% file: discussion.tex
\modify{
The designed model patterns support important functionalities in
modeling systems with statecharts, i.e., the two-way communication
functionality and the configurable execution order functionality.
They can be directly applied to any application scenario that requires
two-way communication or configurable execution order.
If a statechart model needs to apply multiple model patterns,
we can just apply these model patterns one by one. For example,
the airway laser surgery statechart model shown in Fig.~\ref{fig:laserPattern}
first applies the two-way communication model pattern and then applies
the configurable execution order model pattern according to
Procedure~\ref{prc:applyTWC} and Procedure~\ref{prc:applyEOM}, respectively.
Given a statechart model with $N$ transitions, the time complexity of
applying the model patterns is $O(N)$.

The model patterns contain two major parts: a $\mathtt{Manager}$
statechart and a \textit{interface}.
If we want to implement the model patterns in other statechart-based modeling
platforms such as Stateflow~\cite{stateflow}, we only need to re-implement
the \textit{interface} based on corresponding modeling platforms' features.
The $\mathtt{Manager}$ statechart does not need to be re-designed,
as it only contains basic statechart elements, i.e., \textit{state},
\textit{transition}, \textit{guard}, and \textit{action}.
The proposed model pattern design approach is generalized and can be
applied to facilitate other functionalities that
are not directly supported by statecharts, such as exception handling.

The correctness of the model patterns is proved in Theorem~\ref{thm:TWC}
and Theorem~\ref{thm:EOM}. We also use the simplified airway laser surgery case
study to illustrate the correctness of the model patterns when they are accompanied
by medical guideline models.
To ensure the completeness of the model patterns' correctness proof,
we need to verify below properties when the model patterns are
applied in a statechart model.
For the two-way communication model pattern, we need to verify
that every \textit{event} sent by a statechart can be received
by all other statecharts in the model.
Regarding the configurable execution order model pattern,
for all possible statechart execution order configurations,
we need to verify that every statechart's initial state is reached
in sequence as configured.
}

%% file: conclusion.tex
Some essential functionalities in medical operations, such as two-way communication
and configurable execution order, are not directly supported by some open
source statechart modeling tools.
The paper presents an approach to apply model patterns to support these
essential functionalities and formally prove the correctness of
designed model patterns.
The approach can be applied to facilitate other functionalities that
are not directly supported by statecharts, such as exception handling.
The designed model patterns can be directly applied in many
application domains in which these functionalities are needed.

%% file: appendix.tex
\section{Proof of the Lemmas}

\subsection{Proof of Lemma~\ref{lm:initEventQueue}}
\label{subsec:initEventQueueProof}

\begin{proof}
%	According to the composition rule~\cite{book536}, the proof of the lemma is
%	equivalent to prove the following two triples:
%	\begin{align}
%	\label{eq:initS1}
%	\triple{P}{S_1;a:=c}{R}	
%	\end{align}	
%	\begin{align}
%	\label{eq:initS2}
%	\triple{R}{S_2;x:=c}{Q}	
%	\end{align}
%	where $P \equiv c < \mathtt{stNum} \land \mathtt{stNum}>1$,
%	$Q \equiv (x=1 \land \mathtt{exe}=\mathtt{true} \land n=0) \lor ((x=0 \lor x=c+1) \land \mathtt{exe}=\mathtt{false})$,
%	$R \equiv (a=c \land a=0 \land \mathtt{exe}=\mathtt{true} \land n=0 \land \mathtt{stNum}>1) \ \lor \ (a=c \land a < \mathtt{stNum} \land \mathtt{exe}=\mathtt{false} \land \mathtt{stNum}>1)$,
%	\begin{align*}
%	S_1 \equiv
%	& \quad \kwif c=0 \ \kwthen \ \mathtt{exe}:=\mathtt{true}; \ n:=0; \\
%	& \quad \kwelse \mathtt{exe}:=\mathtt{false}; \ \kwfi
%	\end{align*}
%	and
%	\begin{align*}
%	S_2 \equiv
%	& \quad \kwif c=\mathtt{stNum}-1 \ \kwthen \ c:=0; \\
%	& \quad \kwelse c:=a+1; \ \kwfi
%	\end{align*}
	
	The proof of triple~\eqref{eq:initS2} is as follows.
	\begin{align*}
	& \quad Q [x:=c] \\
	& \equiv (x=1 \land \mathtt{exe}=\mathtt{true} \land n=0) \ \lor \\
	& \quad ((x=0 \lor x=c+1) \land \mathtt{exe}=\mathtt{false}) [x:=c] \\
	& \equiv (c=1 \land \mathtt{exe}=\mathtt{true} \land n=0) \ \lor \\
	& \quad ((c=0 \lor c=c+1) \land \mathtt{exe}=\mathtt{false}) \\
	& \equiv (c=1 \land \mathtt{exe}=\mathtt{true} \land n=0) \ \lor \\
	& \quad (c=0 \land \mathtt{exe}=\mathtt{false}) \\
	& \equiv Q_1
	\end{align*}
	To prove $\triple{R}{S_2}{Q_1}$ is equivalent to prove the
	following two triples:
	\begin{align}
	\label{eq:initS2if}
	\triple{R \land c=\mathtt{stNum}-1}{c:=0}{Q_1}	
	\end{align}	
	\begin{align}
	\label{eq:initS2else}
	\triple{R \land c\ne\mathtt{stNum}-1}{c:=a+1}{Q_1}	
	\end{align}
	
	The proof of triple~\eqref{eq:initS2if} is as follows.
	\begin{align*}
	& \quad Q_1 [c:=0] \\
	& \equiv (c=1 \land \mathtt{exe}=\mathtt{true} \land n=0) \ \lor \\
	& \quad (c=0 \land \mathtt{exe}=\mathtt{false}) [c:=0] \\
	& \equiv (0=1 \land \mathtt{exe}=\mathtt{true} \land n=0) \ \lor \\
	& \quad (0=0 \land \mathtt{exe}=\mathtt{false}) \\
	& \equiv \mathtt{exe}=\mathtt{false} \\
	& \equiv Q_2
	\end{align*}	
	\begin{align*}
	& \quad R \land c=\mathtt{stNum}-1 \\
	& \equiv ((a=c \land a=0 \land \mathtt{exe}=\mathtt{true} \land n=0 \land \mathtt{stNum}>1) \ \lor \\
	& \quad (a=c \land a < \mathtt{stNum} \land \mathtt{exe}=\mathtt{false}) \land \mathtt{stNum}>1) \ \land \\
	& \quad c=\mathtt{stNum}-1 \\
	& \equiv \mathtt{false} \ \lor \ (a=c \land c=\mathtt{stNum}-1 \land \mathtt{exe}=\mathtt{false} \ \land \\
	& \quad \mathtt{stNum}>1) \\
	& \equiv a=c \land c=\mathtt{stNum}-1 \land \mathtt{exe}=\mathtt{false} \land \mathtt{stNum}>1
	\end{align*}
	As $R \land c=\mathtt{stNum}-1 \rightarrow Q_2$, hence the triple~\eqref{eq:initS2if} is correct.
	
	The proof of triple~\eqref{eq:initS2else} is as follows.
	\begin{align*}
	& \quad Q_1 [c:=a+1] \\
	& \equiv (c=1 \land \mathtt{exe}=\mathtt{true} \land n=0) \ \lor \\
	& \quad (c=0 \land \mathtt{exe}=\mathtt{false}) [c:=a+1] \\
	& \equiv (a+1=1 \land \mathtt{exe}=\mathtt{true} \land n=0) \ \lor \\
	& \quad (a+1=0 \land \mathtt{exe}=\mathtt{false}) \\
	& \equiv (a=0 \land \mathtt{exe}=\mathtt{true} \land n=0) \ \lor \\
	& \quad (a=-1 \land \mathtt{exe}=\mathtt{false}) \\
	& \equiv Q_3
	\end{align*}	
	\begin{align*}
	& \quad R \land c \ne \mathtt{stNum}-1 \\
	& \equiv ((a=c \land a=0 \land \mathtt{exe}=\mathtt{true} \land n=0 \land \mathtt{stNum}>1) \ \lor \\
	& \quad (a=c \land a < \mathtt{stNum} \land \mathtt{exe}=\mathtt{false} \land \mathtt{stNum}>1)) \ \land \\
	& \quad c \ne \mathtt{stNum}-1 \\
	& \equiv (a=c \land a=0 \land \mathtt{exe}=\mathtt{true} \land n=0 \land \mathtt{stNum}>1) \ \lor \\
	& \quad (a=c \land a < \mathtt{stNum}-1 \land \mathtt{exe}=\mathtt{false} \land \mathtt{stNum}>1) \\
	\end{align*}
	As $R \land c \ne \mathtt{stNum}-1 \rightarrow Q_3$, hence the triple~\eqref{eq:initS2else} is correct. Therefore, the triple~\eqref{eq:initS2} is correct.

	Similarly, the proof of triple~\eqref{eq:initS1} is as follows.
	\begin{align*}
	& \quad R [a:=c] \\
	& \equiv (a=c \land a=0 \land \mathtt{exe}=\mathtt{true} \land n=0 \land \mathtt{stNum}>1) \ \lor \\
	& \quad (a=c \land a < \mathtt{stNum} \land \mathtt{exe}=\mathtt{false} \land \mathtt{stNum}>1) [a:=c] \\	
	& \equiv (c=0 \land \mathtt{exe}=\mathtt{true} \land n=0 \land \mathtt{stNum}>1) \ \lor \\
	& \quad (c < \mathtt{stNum} \land \mathtt{exe}=\mathtt{false} \land \mathtt{stNum}>1) \\
	& \equiv R_1
	\end{align*}
	To prove $\triple{P}{S_1}{R_1}$ is equivalent to prove the
	following two triples:
	\begin{align}
	\label{eq:initS1if}
	\triple{P \land c=0}{\mathtt{exe}:=\mathtt{true};n:=0}{R_1}	
	\end{align}	
	\begin{align}
	\label{eq:initS1else}
	\triple{P \land c \ne 0}{\mathtt{exe}:=\mathtt{false}}{R_1}	
	\end{align}
	
	The proof of triple~\eqref{eq:initS1if} is as follows.
	\begin{align*}
	& \quad R_1 [n:=0][\mathtt{exe}:=\mathtt{true}] \\	
	& \equiv (c=0 \land \mathtt{exe}=\mathtt{true} \land n=0 \land \mathtt{stNum}>1) \lor (c < \mathtt{stNum} \\
	& \quad \land \mathtt{exe}=\mathtt{false} \land \mathtt{stNum}>1) [n:=0][\mathtt{exe}:=\mathtt{true}] \\	
	& \equiv (c=0 \land \mathtt{exe}=\mathtt{true} \land \mathtt{stNum}>1) \ \lor \\
	& \quad (c < \mathtt{stNum} \land \mathtt{exe}=\mathtt{false} \land \mathtt{stNum}>1) [\mathtt{exe}:=\mathtt{true}] \\	
	& \equiv c=0 \land \mathtt{stNum}>1 \\
	& \equiv R_2
	\end{align*}	
	\begin{align*}
	& \quad P \land c=0 \\
	& \equiv c < \mathtt{stNum} \land \mathtt{stNum}>1 \land c=0 \\
	& \equiv c=0 \land \mathtt{stNum}>1
	\end{align*}
	As $P \land c=0 \equiv R_2$, hence the triple~\eqref{eq:initS1if} is correct.
	
	The proof of triple~\eqref{eq:initS1else} is as follows.
	\begin{align*}
	& \quad R_1 [\mathtt{exe}:=\mathtt{false}] \\
	& \equiv (c=0 \land \mathtt{exe}=\mathtt{true} \land n=0 \land \mathtt{stNum}>1) \ \lor \\
	& \quad (c < \mathtt{stNum} \land \mathtt{exe}=\mathtt{false} \land \mathtt{stNum}>1) [\mathtt{exe}:=\mathtt{false}] \\	
	& \equiv c < \mathtt{stNum} \land \mathtt{stNum}>1 \\
	& \equiv R_3
	\end{align*}	
	\begin{align*}
	& \quad P \land c \ne 0 \\
	& \equiv c < \mathtt{stNum} \land \mathtt{stNum}>1 \land c \ne 0
	\end{align*}
	As $P \land c \ne 0 \rightarrow R_3$, hence the triple~\eqref{eq:initS1else} is correct.
	Therefore, the triple~\eqref{eq:initS1} is correct.
	
	Therefore, the lemma is correct.
%	\qed
\end{proof}

\subsection{Proof of Lemma~\ref{lm:pop}}
\label{subsec:popProof}

\begin{proof}
%	The loop invariant and bound function of the $\kwwhile$ loop
%	in the $\mathtt{pop}$ program are
%	$\mathtt{\textbf{inv}} \equiv 0 \le i \le n \land (x=\mathtt{false} \lor (x=\mathtt{true} \land E[v]=e \land ((\mathtt{exe}=\mathtt{true} \land r>S[v]) \lor (\mathtt{exe}=\mathtt{false} \land r<S[v]))))$
%	and
%	$\mathtt{\textbf{bd}} \equiv n-i$, respectively.
%	
%	We use $P$ and $Q$ to denote the precondition and postcondition
%	of the triple, i.e.,
%	$P \equiv e>0 \land r>0$ and
%	$Q \equiv x=\mathtt{false} \lor (x=\mathtt{true} \land E[v]=e \land ((\mathtt{exe}=\mathtt{true} \land r>S[v]) \lor (\mathtt{exe}=\mathtt{false} \land r<S[v])))$.
%	We prove the lemma in the following five steps.
	
	~\newline
	\textbf{Step 1:} prove that the initialization establishes
	the loop invariant, i.e., $\triple{P}{x:=\mathtt{false}; i:=0; v:=0}{\mathtt{inv}}$.
	
	The proof of triple $\triple{P}{x:=\mathtt{false}; i:=0; v:=0}{\mathtt{inv}}$
	is as follows.	
	\begin{align*}
	& \quad \ \mathtt{inv}[v:=0][i:=0][x:=\mathtt{false}]\\	
	& \equiv 0 \le i \le n \land (x=\mathtt{false} \lor (x=\mathtt{true} \ \land \\
	& \quad \ E[v]=e \ \land \ ((\mathtt{exe}=\mathtt{true} \land r>S[v]) \ \lor \\
	& \quad \ (\mathtt{exe}=\mathtt{false} \land r<S[v])))) [v:=0][i:=0][x:=\mathtt{false}] \\	
	& \equiv 0 \le i \le n \land (x=\mathtt{false} \lor (x=\mathtt{true} \ \land \\
	& \quad \ E[0]=e \ \land \ ((\mathtt{exe}=\mathtt{true} \land r>S[0]) \ \lor \\
	& \quad \ (\mathtt{exe}=\mathtt{false} \land r<S[0])))) [i:=0][x:=\mathtt{false}] \\	
	& \equiv \mathtt{true} \land (x=\mathtt{false} \lor (x=\mathtt{true} \ \land \\
	& \quad \ E[0]=e \ \land \ ((\mathtt{exe}=\mathtt{true} \land r>S[0]) \ \lor \\
	& \quad \ (\mathtt{exe}=\mathtt{false} \land r<S[0])))) [x:=\mathtt{false}] \\	
	& \equiv \mathtt{true} \land (\mathtt{true} \lor (\mathtt{false} \ \land \\
	& \quad \ E[0]=e \ \land \ ((\mathtt{exe}=\mathtt{true} \land r>S[0]) \ \lor \\
	& \quad \ (\mathtt{exe}=\mathtt{false} \land r<S[0])))) \\
	& \equiv \mathtt{true}		
	\end{align*}
	As $P \rightarrow \mathtt{true}$, hence the triple
	$\triple{P}{x:=\mathtt{false}; i:=0; v:=0}{\mathtt{inv}}$
	is correct.

	~\newline
	\textbf{Step 2:} prove that the loop body does not change
	the loop invariant, i.e., $\triple{\mathtt{inv} \land i < n}{S}{\mathtt{inv}}$,
	where
	\begin{align*}
	S \equiv	
	& \quad \kwif E[i] = e \land (\mathtt{exe}=\mathtt{true} \land r > S[i]) \ \lor \\
	& \quad \qquad (\mathtt{exe}=\mathtt{false} \land r < S[i]) \ \kwthen \\
	& \quad \qquad v:=i; \ x := \mathtt{true}; \ \kwfi; \\
	& \quad i:=i+1
	\end{align*}
	
	According to the composition rule~\cite{book536}, the proof
	of $\triple{\mathtt{inv} \land i < n}{S}{\mathtt{inv}}$ is
	equivalent to prove the following two triples:
	\begin{align}
	\label{eq:popStep2S1}
	\triple{\mathtt{inv} \land i < n}{S_1}{R}	
	\end{align}
	\begin{align}
	\label{eq:popStep2S2}
	\triple{R}{i:=i+1}{\mathtt{inv}}	
	\end{align}
	where	
	$R \equiv -1 \le i \le n-1 \land (x=\mathtt{false} \lor (x=\mathtt{true} \land E[v]=e \land ((\mathtt{exe}=\mathtt{true} \land r>S[v]) \lor (\mathtt{exe}=\mathtt{false} \land r<S[v]))))$,
	\begin{align*}
	S_1 \equiv	
	& \quad \kwif B \ \kwthen \ v:=i; \ x := \mathtt{true}; \ \kwfi
	\end{align*}
	and
	$B \equiv E[i] = e \land ((\mathtt{exe}=\mathtt{true} \land r > S[i]) \lor (\mathtt{exe}=\mathtt{false} \land r < S[i]))$.	
	
	The proof of triple~\eqref{eq:popStep2S2} is as follows.
	\begin{align*}
	& \quad \ \mathtt{inv} [i:=i+1] \\	
	& \equiv 0 \le i \le n \land (x=\mathtt{false} \lor (x=\mathtt{true} \ \land \\
	& \quad \ E[v]=e \ \land \ ((\mathtt{exe}=\mathtt{true} \land r>S[v]) \ \lor \\
	& \quad \ (\mathtt{exe}=\mathtt{false} \land r<S[v])))) [i:=i+1] \\	
	& \equiv -1 \le i \le n-1 \land (x=\mathtt{false} \lor (x=\mathtt{true} \ \land \\
	& \quad \ E[v]=e \ \land \ ((\mathtt{exe}=\mathtt{true} \land r>S[v]) \ \lor \\
	& \quad \ (\mathtt{exe}=\mathtt{false} \land r<S[v])))) \\
	& \equiv R
	\end{align*}
	Hence, the triple~\eqref{eq:popStep2S2} is correct.
	
	The proof of triple~\eqref{eq:popStep2S1}
	is equivalent to prove the
	following two triples:
	\begin{align}
	\label{eq:popStep2S1if}
	\triple{\mathtt{inv} \land i < n \land B}{v:=i; x:=\mathtt{true}}{R}	
	\end{align}	
	\begin{align}
	\label{eq:popStep2S1else}
	\triple{\mathtt{inv} \land i < n \land \neg B}{\kwskip}{R}	
	\end{align}
	
	The proof of triple~\eqref{eq:popStep2S1if} is as follows.
	\begin{align*}
	& \quad \ R [x:=\mathtt{true}] [v:=i] \\	
	& \equiv -1 \le i \le n-1 \land (x=\mathtt{false} \lor (x=\mathtt{true} \land \\
	& \quad \ E[v]=e \ \land \ ((\mathtt{exe}=\mathtt{true} \land r>S[v]) \ \lor \\
	& \quad \ (\mathtt{exe}=\mathtt{false} \land r<S[v])))) [x:=\mathtt{true}] [v:=i] \\	
	& \equiv -1 \le i \le n-1 \land (E[v]=e \land ((\mathtt{exe}=\mathtt{true} \land r>S[v]) \\
	& \quad \ \lor (\mathtt{exe}=\mathtt{false} \land r<S[v]))) [v:=i] \\	
	& \equiv -1 \le i \le n-1 \land (E[i]=e \land ((\mathtt{exe}=\mathtt{true} \land r>S[i]) \\
	& \quad \ \lor (\mathtt{exe}=\mathtt{false} \land r<S[i]))) \\
	& \equiv R_1
	\end{align*}	
	\begin{align*}
	& \quad \ \mathtt{inv} \land i < n \land B \\	
	& \equiv 0 \le i \le n \land (x=\mathtt{false} \lor (x=\mathtt{true} \ \land \\
	& \quad \ E[v]=e \ \land \ ((\mathtt{exe}=\mathtt{true} \land r>S[v]) \ \lor \\
	& \quad \ (\mathtt{exe}=\mathtt{false} \land r<S[v])))) \land i<n \ \land \\
	& \quad \ E[i] = e \land (\mathtt{exe}=\mathtt{true} \land r > S[i]) \ \lor \\
	& \quad \ (\mathtt{exe}=\mathtt{false} \land r < S[i]) \\	
	& \equiv 0 \le i < n \land (x=\mathtt{false} \lor (x=\mathtt{true} \ \land \\
	& \quad \ E[v]=e \ \land \ ((\mathtt{exe}=\mathtt{true} \land r>S[v]) \ \lor \\
	& \quad \ (\mathtt{exe}=\mathtt{false} \land r<S[v])))) \ \land \\
	& \quad \ E[i] = e \land (\mathtt{exe}=\mathtt{true} \land r > S[i]) \ \lor \\
	& \quad \ (\mathtt{exe}=\mathtt{false} \land r < S[i])
	\end{align*}	
	As $i$ is an integer, hence $0 \le i < n \rightarrow -1 \le i \le n-1$.
	As $\mathtt{inv} \land i < n \land B \rightarrow R_1$,
	hence the triple~\eqref{eq:popStep2S1if} is correct.
	
	Similarly, we prove the triple~\eqref{eq:popStep2S1else} as follows.
	\begin{align*}
	& \quad \ \mathtt{inv} \land i < n \land \neg B \\	
	& \equiv 0 \le i \le n \land (x=\mathtt{false} \lor (x=\mathtt{true} \ \land \\
	& \quad \ E[v]=e \ \land \ ((\mathtt{exe}=\mathtt{true} \land r>S[v]) \ \lor \\
	& \quad \ (\mathtt{exe}=\mathtt{false} \land r<S[v])))) \land i<n \land \neg B \\	
	& \equiv 0 \le i < n \land (x=\mathtt{false} \lor (x=\mathtt{true} \ \land \\
	& \quad \ E[v]=e \ \land \ ((\mathtt{exe}=\mathtt{true} \land r>S[v]) \ \lor \\
	& \quad \ (\mathtt{exe}=\mathtt{false} \land r<S[v])))) \land i<n \land \neg B
	\end{align*}
	As $\mathtt{inv} \land i < n \land \neg B \rightarrow R$,
	hence the triple~\eqref{eq:popStep2S1else} is correct.
	Therefore, the triple~\eqref{eq:popStep2S1} and triple	
	$\triple{\mathtt{inv} \land i < n}{S}{\mathtt{inv}}$ are correct.

	~\newline
	\textbf{Step 3:} prove that the bound function decreases after each iteration,
	i.e., $\triple{\mathtt{inv} \land i < n \land \mathtt{bd} = z}{S}{\mathtt{bd} < z}$,
	where $z$ is an integer.
	
	According to the composition rule~\cite{book536}, the proof
	of $\triple{\mathtt{inv} \land i < n \land \mathtt{bd} = z}{S}{\mathtt{bd} < z}$ is
	equivalent to prove the following two triples:
	\begin{align}
	\label{eq:popStep3S1}
	\triple{\mathtt{inv} \land i < n \land \mathtt{bd} = z}{S_1}{n-i<z+1}	
	\end{align}
	\begin{align}
	\label{eq:popStep3S2}
	\triple{n-i<z+1}{i:=i+1}{\mathtt{bd} < z}	
	\end{align}	
	
	The proof of triple~\eqref{eq:popStep3S2} is as follows.
	\begin{align*}
	& \quad \ \mathtt{bd} < z [i:=i+1] \\	
	& \equiv n-i<z [i:=i+1] \\	
	& \equiv n-i<z+1
	\end{align*}
	Hence, the triple~\eqref{eq:popStep3S2} is correct.
	
	The proof of triple~\eqref{eq:popStep3S1}
	is equivalent to prove the
	following two triples:
	\begin{align}
	\label{eq:popStep3S1if}
	\triple{\mathtt{inv} \land i < n \land \mathtt{bd} = z \land B}{v:=i; x:=\mathtt{true}}{n-i<z+1}	
	\end{align}	
	\begin{align}
	\label{eq:popStep3S1else}
	\triple{\mathtt{inv} \land i < n \land \mathtt{bd} = z \land \neg B}{\kwskip}{n-i<z+1}	
	\end{align}
	
	The proof of triple~\eqref{eq:popStep3S1if} is as follows.
	\begin{align*}
	& \quad \ n-i<z+1 [x:=\mathtt{true}] [v:=i] \\	
	& \equiv n-i<z+1
	\end{align*}	
	\begin{align*}
	& \quad \ \mathtt{inv} \land i < n \land \mathtt{bd} = z \land B \\	
	& \equiv n-i=z \land \mathtt{inv} \land i < n \land B
	\end{align*}	
	As $\mathtt{inv} \land i < n \land \mathtt{bd} = z \land B \rightarrow n-i<z+1$,
	hence the triple~\eqref{eq:popStep3S1if} is correct.
	
	Similarly, as $\mathtt{inv} \land i < n \land \mathtt{bd} = z \land \neg B  \rightarrow n-i<z+1$,
	hence the triple~\eqref{eq:popStep3S1else} is correct.
	Therefore, the triple~\eqref{eq:popStep3S1} and triple	
	$\triple{\mathtt{inv} \land i < n \land \mathtt{bd} = z}{S}{\mathtt{bd} < z}$
	are correct.

	~\newline	
	\textbf{Step 4:} prove that the loop invariant implies the bound
	function is non-negative, i.e., $\mathtt{inv} \rightarrow \mathtt{bd} \ge 0$.
	
	As $\mathtt{bd} \ge 0 \equiv n \ge i$, hence $\mathtt{inv} \rightarrow \mathtt{bd} \ge 0$.	
	
	~\newline
	\textbf{Step 5:} prove that the postcondition holds when the
	$\kwwhile$ loop terminates, i.e., $\mathtt{inv} \land \neg (i < n) \rightarrow Q$.
	
	\begin{align*}
	& \quad \ \mathtt{inv} \land \neg (i < n) \\	
	& \equiv 0 \le i \le n \land (x=\mathtt{false} \lor (x=\mathtt{true} \land \\
	& \quad \ E[v]=e \ \land \ ((\mathtt{exe}=\mathtt{true} \land r>S[v]) \ \lor \\
	& \quad \ (\mathtt{exe}=\mathtt{false} \land r<S[v])))) \ \land i \ge n \\
	& \equiv i=n \land (x=\mathtt{false} \lor (x=\mathtt{true} \land E[v]=e \land \\
	& \quad \ ((\mathtt{exe}=\mathtt{true} \land r>S[v]) \lor (\mathtt{exe}=\mathtt{false} \land r<S[v])))) \\
	& \rightarrow Q
	\end{align*}
%	\qed	
\end{proof}

\subsection{Proof of Lemma~\ref{lm:updateExeInfo}}
\label{subsec:updateExeInfoProof}

\begin{proof}
%	According to the composition rule~\cite{book536}, the proof of the lemma is
%	equivalent to prove the following two triples:
%	\begin{align}
%	\label{eq:updateExeInfoS1}
%	\triple{P}{S}{R}	
%	\end{align}
%	\begin{align}
%	\label{eq:updateExeInfoS2}
%	\triple{R}{x:=a}{Q}	
%	\end{align}
%	where $P \equiv t>0 \land \mathtt{stNum}>0$,
%	$Q \equiv (t=\mathtt{stNum} \land x=1) \lor (t \ne \mathtt{stNum} \land x=t+1)$,
%	$R \equiv (t=\mathtt{stNum} \land a=1) \lor (t \ne \mathtt{stNum} \land a=t+1)$,
%	and
%	\begin{align*}
%	S \equiv
%	& \quad \kwif t = \mathtt{stNum} \ \kwthen \ a := 1; \\
%	& \quad \kwelse a := t + 1; \ \kwfi
%	\end{align*}
	
	The proof of triple~\eqref{eq:updateExeInfoS2} is as follows.
	\begin{align*}
	& \quad \ Q [x := a] \\
	& \equiv (t=\mathtt{stNum} \land x=1) \lor (t \ne \mathtt{stNum} \land x=t+1) [x := a] \\
	& \equiv (t=\mathtt{stNum} \land a=1) \lor (t \ne \mathtt{stNum} \land a=t+1) \\
	& \equiv R
	\end{align*}
	Hence, the triple~\eqref{eq:updateExeInfoS2} is correct.

	The proof of triple~\eqref{eq:updateExeInfoS1} is
	equivalent to prove the following two triples:
	\begin{align}
	\label{eq:updateExeInfoS1if}
	\triple{P \land t = \mathtt{stNum}}{a := 1}{R}	
	\end{align}	
	\begin{align}
	\label{eq:updateExeInfoS1else}
	\triple{P \land t \ne \mathtt{stNum}}{a := t + 1}{R}	
	\end{align}
	
	The proof of triple~\eqref{eq:updateExeInfoS1if} is as follows.
	\begin{align*}
	& \quad \ R [a := 1] \\	
	& \equiv (t=\mathtt{stNum} \land a=1) \lor (t \ne \mathtt{stNum} \land a=t+1) [a := 1] \\	
	& \equiv t=\mathtt{stNum} \lor (t \ne \mathtt{stNum} \land t=0) \\
	& \equiv R_1
	\end{align*}
	As $P \land t = \mathtt{stNum} \rightarrow R_1$, hence the triple~\eqref{eq:updateExeInfoS1if} is correct.
	
	The proof of triple~\eqref{eq:updateExeInfoS1else} is as follows.
	\begin{align*}
	& \quad \ R [a := t + 1] \\
	& \equiv (t=\mathtt{stNum} \land a=1) \lor (t \ne \mathtt{stNum} \land a=t+1) [a := t + 1] \\	
	& \equiv (t=\mathtt{stNum} \land t=0) \lor t \ne \mathtt{stNum} \\
	& \equiv R_2
	\end{align*}	
	As $P \land t \ne \mathtt{stNum} \rightarrow R_2$, hence the triple~\eqref{eq:updateExeInfoS1else} is correct.
	Therefore, the triple~\eqref{eq:updateExeInfoS1} is correct.
	
	Therefore, the lemma is correct.
%	\qed
\end{proof}

\subsection{Proof of Lemma~\ref{lm:run}}
\label{subsec:runProof}

\begin{proof}	
%	The proof of the lemma is equivalent to prove the following two triples:
%	\begin{align}
%	\label{eq:runS1}
%	\triple{P \land O[t-1] = \mathtt{st}}{x := \mathtt{true}}{Q}	
%	\end{align}	
%	\begin{align}
%	\label{eq:runS2}
%	\triple{P \land O[t-1] \ne \mathtt{st}}{x := \mathtt{false}}{Q}	
%	\end{align}
%	where $P \equiv t>0 \land \mathtt{st} > 0$ and 
%	$Q \equiv (x=\mathtt{true} \land O[t-1] = \mathtt{st}) \lor (x=\mathtt{false} \land O[t-1] \ne \mathtt{st})$.
	
	The proof of triple~\eqref{eq:runS1} is as follows.
	\begin{align*}
	& \quad Q [x := \mathtt{true}] \\
	& \equiv (x=\mathtt{true} \land O[t-1] = \mathtt{st}) \ \lor \\
	& \quad \ (x=\mathtt{false} \land O[t-1] \ne \mathtt{st}) [x := \mathtt{true}] \\	
	& \equiv (\mathtt{true}=\mathtt{true} \land O[t-1] = \mathtt{st}) \ \lor \\
	& \quad \ (\mathtt{true}=\mathtt{false} \land O[t-1] \ne \mathtt{st}) \\	
	& \equiv O[t-1] = \mathtt{st}
	\end{align*}
	As $P \land O[t-1] = \mathtt{st} \rightarrow O[t-1] = \mathtt{st}$,
	hence the triple~\eqref{eq:runS1} is correct.
	
	Similarly, the proof of triple~\eqref{eq:runS2} is as follows.
	\begin{align*}
	& \quad Q [x := \mathtt{false}] \\
	& \equiv (x=\mathtt{true} \land O[t-1] = \mathtt{st}) \ \lor \\
	& \quad \ (x=\mathtt{false} \land O[t-1] \ne \mathtt{st}) [x := \mathtt{false}] \\	
	& \equiv (\mathtt{false}=\mathtt{true} \land O[t-1] = \mathtt{st}) \ \lor \\
	& \quad \ (\mathtt{false}=\mathtt{false} \land O[t-1] \ne \mathtt{st}) \\	
	& \equiv O[t-1] \ne \mathtt{st}
	\end{align*}
	As $P \land O[t-1] \ne \mathtt{st} \rightarrow O[t-1] \ne \mathtt{st}$,
	hence the triple~\eqref{eq:runS2} is correct.
	
	Therefore, the lemma is correct.
%	\qed
\end{proof}

%% file: acknowledgement.tex
This work is supported in part by NSF CNS 1545008,
NSF CNS 1842710, and NSF CNS 1545002.